%% file: paper.tex
\documentclass[letterpaper]{article}
\usepackage{hyperref} 
\pdfoutput=1 
\usepackage{ijcai15}
\usepackage{times}
\usepackage{latexsym}
\usepackage{url}
\usepackage{amsmath}
\allowdisplaybreaks
\usepackage{amsthm}
\usepackage{thm-restate}
\usepackage{amsfonts}
\usepackage{amssymb}
\usepackage{macros}
\usepackage[draft]{ifdraft}
\newtheorem{claim}{Claim}

\title{Combining Rewriting and Incremental Materialisation Maintenance \\ for Datalog Programs with Equality}
\author{
   Boris Motik, Yavor Nenov, Robert Piro and Ian Horrocks\\
   Department of Computer Science, Oxford University\\
   Oxford, United Kingdom\\
   firstname.lastname@cs.ox.ac.uk
}

\begin{document}

\maketitle

\begin{abstract}
\input{abstract}
\end{abstract}

\input{introduction}
\input{preliminaries}
\input{incremental}
\input{evaluation}
\input{conclusion}

\section*{Acknowledgments}

This work was funded by the EPSRC projects MaSI$^3$, Score!, and DBOnto, and
the FP7 project Optique.

\bibliographystyle{named}
\bibliography{references}

\ifdraft{
    \clearpage
    \appendix
    \onecolumn
    \input{appendix}
}{}

\end{document}

%% file: abstract.tex
\emph{Materialisation} precomputes all consequences of a set of facts and a
datalog program so that queries can be evaluated directly (i.e., independently
from the program). \emph{Rewriting} optimises materialisation for datalog
programs with equality by replacing all equal constants with a single
representative; and \emph{incremental maintenance} algorithms can efficiently
update a materialisation for small changes in the input facts. Both techniques
are critical to practical applicability of datalog systems; however, we are
unaware of an approach that combines rewriting and incremental maintenance. In
this paper we present the first such combination, and we show empirically that
it can speed up updates by several orders of magnitude compared to using either
rewriting or incremental maintenance in isolation.

%% file: introduction.tex
\section{Introduction}\label{sec:introduction}

\emph{Datalog} \cite{abiteboul95foundation} is a declarative, rule-based
language that can describe (possibly recursive) data dependencies. It is widely
used in applications as diverse as enterprise data management
\cite{DBLP:conf/iclp/Aref10} and query answering over ontologies in the OWL 2
RL profile \cite{owl2-profiles} extended with SWRL rules \cite{SWRL}.

Querying the set $\Pi^\infty(E)$ of consequences of a set of \emph{explicit}
facts $E$ and a datalog program $\Pi$ is a key service in datalog systems. It
can be supported by precomputing and storing $\Pi^\infty(E)$ so that queries
can be evaluated directly, without further reference to $\Pi$. Set
$\Pi^\infty(E)$ and the process of computing it are called the
\emph{materialisation} of $E$ w.r.t.\ $\Pi$. This technique is used in the
state of the art systems such as Olwgres \cite{DBLP:conf/owled/StockerS08},
WebPIE \cite{DBLP:journals/ws/UrbaniKMHB12}, Oracle's RDF store
\cite{DBLP:conf/icde/WuEDCKAS08}, GraphDB (formerly OWLIM)
\cite{DBLP:journals/semweb/BishopKOPTV11}, and RDFox
\cite{mnpho14parallel-materialisation-RDFox}.

Although datalog traditionally employs the unique name assumption (UNA), in
some applications uniqueness of identifiers cannot be guaranteed. For example,
due to the distribution and the independence of data sources, in the Semantic
Web different identifies are often used to refer to the same domain object.
Handling such use cases requires an extension of datalog without UNA, in which
one can infer equalities between constants using a special \emph{equality}
predicate $\sameAs$ that can occur in facts and rule heads. The semantics of
$\sameAs$ can be captured explicitly using rules that \emph{axiomatise}
$\sameAs$ as a congruence relation; however, this is known to be inefficient
when equality is used extensively. Therefore, systems commonly use
\emph{rewriting}
\cite{baader98term,NieuwenhuisRubio:HandbookAR:paramodulation:2001}---an
optimisation where equal constants are replaced with a canonical
\emph{representative}, and only facts containing such representatives are
stored. The benefits of rewriting have been well-documented in practice
\cite{DBLP:conf/icde/WuEDCKAS08,DBLP:journals/ws/UrbaniKMHB12,DBLP:journals/semweb/BishopKOPTV11,mnph15owl-sameAs-rewriting}.

Moreover, datalog applications often need to handle continuous updates to the
set of explicit facts $E$. \emph{Rematerialisation} (i.e., computing the
materialisation from scratch) is often very costly, so \emph{incremental
maintenance} algorithms are often used in practice. Adding facts to $E$ is
trivial as one can simply continue from where the initial materialisation has
finished; hence, given a materialisation $\Pi^\infty(E)$ of $E$ w.r.t.\ $\Pi$
and a set of facts $E^-$, the main challenge for an incremental algorithm is to
efficiently compute ${\Pi^\infty(E \setminus E^-)}$. Several such algorithms
have already been proposed. \emph{Truth maintenance systems}
\cite{DBLP:journals/ai/Doyle79,DBLP:journals/ai/Kleer86,DBLP:conf/edbt/GoasdoueMR13}
track dependencies between facts to efficiently determine whether a fact has a
derivation from ${E \setminus E^-}$, so only facts for which no such
derivations exist are deleted. Such approaches, however, store large amounts of
auxiliary information and are thus often unsuitable for data-intensive
applications. \emph{Counting}
\cite{DBLP:conf/ifip/NicolasY83,DBLP:conf/sigmod/GuptaMS93,DBLP:conf/semweb/UrbaniMJHB13,DBLP:conf/edbt/GoasdoueMR13}
stores with each fact ${F \in \Pi^\infty(E)}$ the number of times $F$ has been
derived during initial materialisation, and this number is used to determine
when to delete $F$; however, in its basic form counting works only with
nonrecursive rules, and a proposed extension to recursive rules requires
multiple counts per fact \cite{DBLP:journals/jiis/DewanOSWS92}, which can be
costly. The \emph{Delete/Rederive} (DRed) algorithm
\cite{DBLP:conf/sigmod/GuptaMS93} handles recursive rules with no storage
overhead: to delete $E^-$ from $E$, the algorithm first overdeletes all
consequences of $E^-$ in $\Pi^\infty(E)$ and then rederives all facts provable
from ${E \setminus E^-}$. The \emph{Backward/Forward} (\BF) algorithm combines
backward and forward chaining in a way that outperforms DRed on inputs where
facts have many alternative derivations---a common scenario in Semantic Web
applications \cite{mnph15incremental-BF}.

Combining rewriting and incremental maintenance is difficult due to complex
interactions between the two techniques: removing $E^-$ from $E$ may entail
retracting equalities, which may (partially) invalidate the rewriting and
require the restoration of rewritten facts (see Section~\ref{sec:incremental}).
To the best of our knowledge, such a combination has not been considered in the
literature, and practical systems either use rewriting with rematerialisation,
or axiomatise equality and use incremental maintenance; in either case they
give up a technique known to be critical for performance. In this paper we
present the \BFeq algorithm, which combines rewriting with \BF: given a set of
facts $E^-$, our algorithm efficiently updates the materialisation of $E$
w.r.t.\ $\Pi$ computed using the rewriting approach by
\citeA{mnph15owl-sameAs-rewriting}. Extensions of datalog with equality are
nowadays used mainly for querying RDF data extended with OWL 2 RL ontologies
and SWRL rules, so we formalise our algorithm in the framework of RDF; however,
our approach can easily be adapted to general datalog.

We have implemented \BFeq in the open-source RDFox
system\footnote{\url{http://www.cs.ox.ac.uk/isg/tools/RDFox/}} and have
evaluated it on several real-world and synthetic datasets. Our results show
that the algorithm indeed combines the best of both worlds, as it is often
several orders of magnitude faster than either rematerialisation with
rewriting, or \BF with axiomatised equality.

%% file: preliminaries.tex
\section{Preliminaries}\label{sec:preliminaries}

\textbf{Datalog.} A \emph{term} is a \emph{constant} (\duri{a}, \duri{b},
\duri{A}, \duri{R}, etc.) or a variable ($x$, $y$, $z$, etc.). An \emph{(RDF)
atom} has the form \triple{t_1}{t_2}{t_3}, where $t_1,t_2,t_3$ are terms; an
\emph{(RDF) fact} (also called a \emph{triple}) is a variable-free RDF atom;
and a \emph{dataset} is a finite set of facts. A \emph{(datalog) rule} $r$ is
an implication of the form \eqref{eq:rule}, where ${H, B_1, \dots, B_n}$ are
atoms and each variable occurring in $H$ also occurs in some $B_i$; $\head{r}
\defeq H$ is the \emph{head atom} of $r$; each $ B_i$ is a \emph{body atom} of
$r$; and $\body{r}$ is the set of all body atoms of $r$. A \emph{(datalog)
program} is a finite set of rules.
\begin{align}
    H \leftarrow B_1 \land \dots \land B_n \label{eq:rule}
\end{align}
A \emph{substitution} is a partial mapping of variables to terms. For $\alpha$
a term, atom, rule, or a set of these, $\voc{\alpha}$ is the set of all
constants in $\alpha$, and $\alpha\sigma$ is the result of applying a
substitution $\sigma$ to $\alpha$. The \emph{materialisation} $\Pi^\infty(E)$
of a dataset $E$ w.r.t.\ a program $\Pi$ is the smallest superset of $E$
containing $\head{r}\sigma$ for each rule ${r \in \Pi}$ and substitution
$\sigma$ with ${\body{r}\sigma \subseteq \Pi^\infty(E)}$.

\interspacing

\noindent\textbf{Equality.} The constant \owl{sameAs} (abbreviated $\sameAs$)
can be used to encode equality between constants. For example, fact
${\triple{\duri{P.\_Smith}}{\sameAs}{\duri{Peter\_Smith}}}$ states that
$\duri{P.\_Smith}$ and $\duri{Peter\_Smith}$ are one and the same object. Facts
of the form $\triple{s}{\sameAs}{t}$ are called \emph{equalities} and, for
readability, are abbreviated as ${s \sameAs t}$; note that ${{\sameAs} \in
\voc{s \sameAs t}}$. Program $\PsameAs$ consisting of rules
\eqref{eq:eq-repl1}--\eqref{eq:eq-ref} axiomatises $\sameAs$ as a congruence
relation. If a program $\Pi$ or a dataset $E$ contain $\sameAs$, systems then
answer queries in the materialisation of $E$ w.r.t.\ ${\Pi \cup \PsameAs}$.
\begin{align}
    \triple{x_1'}{x_2}{x_3} & \leftarrow \triple{x_1}{x_2}{x_3} \wedge x_1 \sameAs x_1'     \tag{$\approx_1$}\label{eq:eq-repl1} \\
    \triple{x_1}{x_2'}{x_3} & \leftarrow \triple{x_1}{x_2}{x_3} \wedge x_2 \sameAs x_2'     \tag{$\approx_2$}\label{eq:eq-repl2} \\
    \triple{x_1}{x_2}{x_3'} & \leftarrow \triple{x_1}{x_2}{x_3} \wedge x_3 \sameAs x_3'     \tag{$\approx_3$}\label{eq:eq-repl3} \\
    x_i \sameAs x_i         & \leftarrow \triple{x_1}{x_2}{x_3}\text{, for }1\leq i\leq 3   \tag{$\approx_4$}\label{eq:eq-ref}
\end{align}

\interspacing

\noindent\textbf{Rewriting} is a well-known optimisation of this approach. For
$\pi$ a mapping of constants to constants and $\alpha$ a constant, fact, rule,
dataset, or substitution, $\pi(\alpha)$ is the result of replacing each
constant $c$ in $\alpha$ with $\pi(c)$; such $\alpha$ is \emph{normal} w.r.t.\
$\pi$ if ${\pi(\alpha) = \alpha}$; and $\pi(\alpha)$ is the
\emph{representative} of $\alpha$ in $\pi$. For $c$ a constant, let ${c^\pi
\defeq \{ d \mid \pi(d) = c \}}$. For $U$ a dataset, let ${U^\pi \defeq \{
\triple{s}{p}{o} \mid \triple{\pi(s)}{\pi(p)}{\pi(o)} \in U \}}$; and, for $F$
a fact, let ${F^\pi \defeq \{ F \}^\pi}$. We assume that all constant are
totally ordered such that $\sameAs$ is the smallest constant; then, for $S$ a
nonempty set of constants, $\min S$ (resp.\ $\max S$) is the smallest (resp.\
greatest) element of $S$. Let $U$ be a dataset and let ${\eqclass{U}{c} \defeq
\{ c \} \cup \{ d \mid c \sameAs d \in U \}}$; then, the \emph{rewriting} of
$U$ is the pair $(\pi,I)$ such that
\begin{enumerate}
    \item ${\pi(c) = \min \eqclass{U}{c}}$ for each constant $c$, and

    \item ${I = \pi(U)}$.
\end{enumerate}
Note that ${\pi(\sameAs) = {\sameAs}}$, that the rewriting is unique for $U$,
and that ${\PsameAs^\infty(U) = U}$ implies ${I^\pi = U}$. The
\emph{r-materialisation} of a dataset $E$ w.r.t.\ a program $\Pi$ is the
rewriting $(\pi,I)$ of the dataset ${J = (\Pi \cup \PsameAs)^\infty(E)}$.
\citeA{mnph15owl-sameAs-rewriting} show how to answer queries over $J$ by
materialising $(\pi,I)$ instead of $J$.

%% file: incremental.tex
\section{Updating R-Materialisation Incrementally}\label{sec:incremental}

Let $E$ and $E^-$ be datasets, let ${E' = E \setminus E^-}$, and let $\Pi$ be a
program. Moreover, let $J$ (resp.\ $J'$) be the materialisation of $E$ (resp.\
$E'$) w.r.t.\ ${\Pi \cup \PsameAs}$, and let $(\pi,I)$ (resp.\ $(\pi',I')$) be
the r-materialisation of $E$ (resp.\ $E'$) w.r.t.\ $\Pi$. Given $(\pi,I)$,
$\Pi$, and $E^-$, the \BFeq algorithm computes $(\pi',I')$ efficiently by
combining the \BF algorithm by \citeA{mnph15incremental-BF} for incremental
maintenance in datalog without equality with the r-materialisation algorithm by
\citeA{mnph15owl-sameAs-rewriting}. We discuss the intuition in
Section~\ref{sec:incremental:intuition} and some optimisations in
Section~\ref{sec:incremental:optimisations}, and we formalise the algorithm in
Section \ref{sec:incremental:formalisation}.

\subsection{Intuition}\label{sec:incremental:intuition}

\noindent\textbf{Main Difficulty.} An update may lead to the deletion of
equalities, which may require \emph{adding} facts to $I$. The following example
program $\Pi$ and dataset $E$ exhibit such behaviour.
\begin{displaymath}
\begin{array}{@{}r@{\,}c@{\,}l@{}}
    \Pi     & = \{  & y_1 \sameAs y_2 \leftarrow \triple{y_1}{\duri{R}}{x} \land \triple{y_2}{\duri{R}}{x}, \\
            &       & y_1 \sameAs y_2 \leftarrow \triple{x}{\duri{R}}{y_1} \land \triple{x}{\duri{R}}{y_2} \,
                \} \\[1ex]
    E       & = \{  & \triple{\duri{a}}{\duri{R}}{\duri{b}}, \,
                      \triple{\duri{c}}{\duri{R}}{\duri{d}}, \,
                      \triple{\duri{a}}{\duri{R}}{\duri{d}}
                \} \\
    I       & = \{  & \triple{\duri{a}}{\duri{R}}{\duri{b}}, \,
                      \duri{a} \sameAs \duri{a}, \,
                      \duri{R} \sameAs \duri{R}, \,
                      \duri{b} \sameAs \duri{b}, \,
                      {\sameAs} \sameAs {\sameAs} \,
                \} \\
    \pi     & = \{  & \duri{a} \mapsto \duri{a}, \,
                      \duri{b} \mapsto \duri{b}, \,
                      \duri{c} \mapsto \duri{a}, \,
                      \duri{d} \mapsto \duri{b}, \,
                      \duri{R} \mapsto \duri{R}, \,
                      {\sameAs} \mapsto {\sameAs} \,
                \} \\[1ex]
    E^-     & = \{  & \triple{\duri{a}}{\duri{R}}{\duri{d}} \,
                \} \\
    I'      & = \{  & \triple{\duri{a}}{\duri{R}}{\duri{b}}, \,
                      \duri{a} \sameAs \duri{a}, \,
                      \duri{R} \sameAs \duri{R}, \,
                      \duri{b} \sameAs \duri{b}, \,
                      {\sameAs} \sameAs {\sameAs}, \\
            &       & \triple{\duri{c}}{\duri{R}}{\duri{d}}, \,
                      \duri{c} \sameAs \duri{c}, \,
                      \duri{d} \sameAs \duri{d} \,
                \} \\
    \pi'    & = \{  & \duri{a} \mapsto \duri{a}, \,
                      \duri{b} \mapsto \duri{b}, \,
                      \duri{c} \mapsto \duri{c}, \,
                      \duri{d} \mapsto \duri{d}, \,
                      \duri{R} \mapsto \duri{R}, \,
                      {\sameAs} \mapsto {\sameAs} \,
                \} \\
\end{array}
\end{displaymath}
Relation $\duri{R}$ is bijective in $\Pi$, so ${\duri{a} \sameAs \duri{c} \in
J}$ as both $\duri{a}$ and $\duri{c}$ have outgoing $\duri{R}$-edges to
$\duri{d}$, and ${\duri{b} \sameAs \duri{d} \in J}$ as both $\duri{b}$ and
$\duri{d}$ have incoming $\duri{R}$-edges from $\duri{a}$. By rewriting, we
represent each fact $\triple{\alpha}{\duri{R}}{\beta}$ from $J$ using a single
fact $\triple{\duri{a}}{\duri{R}}{\duri{b}}$, and analogously for facts
involving $\sameAs$; thus, instead of 14 facts, we store just five facts.
Assume now that we remove $E^-$ from $E$. In $J$ and $J'$ we ascribe no
particular meaning to $\sameAs$, so the monotonicity of datalog ensures ${J
\subseteq J'}$; thus, the \BF algorithm just needs to delete facts that no
longer hold. However, ${\duri{a} \sameAs \duri{c} \not\in J'}$ and ${\duri{b}
\sameAs \duri{d} \not\in J'}$, so we must update $\pi$ and extend $I$ with the
facts from $J'$ that are not represented via $\pi'$. Thus, in our example, $I'$
actually \emph{contains} $I$.

\interspacing

\noindent\textbf{Solution Overview.} \BFeq consists of
Algorithms~\ref{alg:BF-sameAs}--\ref{alg:prove} that follow the same basic idea
as \BF; to highlight the differences, lines that exist in \BF in a modified
form are marked with `$\modsym$', and new lines and algorithms are marked
with`$\newsym$'.

We initially mark all facts in $\pi(E^-)$ as `doubtful'---that is, we indicate
that their truth might change. Next, for each `doubtful' fact $F$, we determine
whether $F$ is provable from $E'$ and, if not, we identify the immediate
consequences of $F$ (i.e., the facts in $I$ that can be derived using $F$) and
mark them as `doubtful'; we know exactly which facts have changed after
processing all `doubtful' facts. To check the provability of $F$, we use
backward chaining to identify the facts in $I$ that can prove $F$, and we use
forward chaining to actually prove $F$. The latter process also identifies the
necessary changes to $\pi$ and $I$, which we apply to $(\pi,I)$ in a final
step. We next describe the components of \BFeq in more detail.

\interspacing

\noindent\textbf{Procedure $\mathsf{saturate}()$} is given a dataset ${C
\subseteq I}$ of \emph{checked} facts, and it computes the set $L$ containing
each fact $F$ derivable from $E'$ such that each fact in a derivation of $F$ is
contained in $C^\pi$; thus, $C$ identifies the part of $J'$ to recompute.
Rather than storing $L$ directly, we adapt the r-materialisation algorithm by
\citeA{mnph15owl-sameAs-rewriting} and represent $L$ by its rewriting
${(\gamma,P \setminus \merged{P})}$; the role of the two sets $P$ and
$\merged{P}$ is discussed shortly.
Lines~\ref{alg:saturate:C}--\ref{alg:saturate:C:expand-F-derive} compute the
facts in $L$ derivable immediately from $E'$: we iterate over each ${F \in C}$
and each ${G \in F^\pi}$; since we represent $L$ by its rewriting, we add
$\gamma(G)$ to $P$. The roles of set $Y$ and
lines~\ref{alg:saturate:C:if-ref}--\ref{alg:saturate:C:if-ref:derive} will be
discussed shortly. Lines~\ref{alg:saturate:F}--\ref{alg:saturate:F:derive}
compute the facts in $L$ derivable using rules: we consider each fact $F$ in
${P \setminus \merged{P}}$
(lines~\ref{alg:saturate:F}--\ref{alg:saturate:add}), each rule $r$, and each
match $\sigma$ of $F$ to a body atom of $r$
(line~\ref{alg:saturate:P:rules:1}), we evaluate the remaining body atoms of
$r$ (line~\ref{alg:saturate:P:rules:2}), and we derive $\gamma(\head{r}\tau)$
for each match $\tau$ (line~\ref{alg:saturate:F:derive}). This basic idea is
slightly more complicated by rewriting: if ${F = a \sameAs b}$, we modify
$\gamma$ so that one constant becomes the representative of the other one
(line~\ref{alg:saturate:RewC}). As a consequence, facts can become `outdated'
w.r.t.\ $\gamma$, so we keep track of such facts using $\merged{P}$: if $F$ is
`outdated', we add $F$ to $\merged{P}$ and $\gamma(F)$ to $P$
(line~\ref{alg:saturate:RewF}); due to the latter, ${P \setminus \merged{P}}$
eventually contains all `up to date' facts. Finally, we apply the reflexivity
rules \eqref{eq:eq-ref} to $F$ (line~\ref{alg:saturate:P:ref}).

Procedure $\mathsf{saturate}()$ is repeatedly called in \BFeq. Set $C$,
however, never shrinks between successive calls, so set $L$ never shrinks
either; hence, at each call we can just continue the computation instead of
starting `from scratch'. A minor problem arises if we derive a fact $F$ with
${F \not\in C^\pi}$ and so we do not add $\gamma(F)$ to $P$, but $C$ is later
extended so that ${F \in C^\pi}$ holds. We handle this by maintaining a set $Y$
of `delayed' facts: in line~\ref{alg:prove:add} we add $F$ to $Y$ if ${F
\not\in C^\pi}$; and in line~\ref{alg:saturate:C:expand-F-derive} we identify
each `delayed' fact ${G \in C^\pi \cap Y}$ and add $\gamma(G)$ to $P$.

\interspacing

\noindent\textbf{Procedure $\mathsf{rewrite}(a,b)$} implements rewriting: we
update $\gamma$ (line~\ref{alg:rewrite:merge}), apply the replacement rules
\eqref{eq:eq-repl1}--\eqref{eq:eq-repl3} to already processed facts containing
`outdated' constants (line~\ref{alg:rewrite:normF}), ensure that $\Gamma$ is
normal w.r.t.\ $\gamma$ (line~\ref{alg:rewrite:normR}), and reapply the
normalised rules (lines \ref{alg:rewrite:reapply}--\ref{alg:rewrite:prove}).
\citeA{mnph15owl-sameAs-rewriting} discuss in detail the issues related to rule
updating and reevaluation.

\interspacing

\noindent\textbf{Procedure $\mathsf{checkProvability()}$} takes a fact ${F \in
I}$ and ensures that, for each ${G \in F^\pi}$, we have ${G \in J'}$ iff
${\gamma(G) \in P \setminus \merged{P}}$---that is, we know the correct status
of each fact that $F$ represents. To this end, we add $F$ to $C$
(line~\ref{alg:checkProvability:add-F}) and thus ensure that $({\gamma,P
\setminus \merged{P}})$ correctly represents $L$
(line~\ref{alg:checkProvability:saturate}). Each fact is added to $C$ only
once, which guarantees termination of the recursion. We then use backward
chaining to examine facts occurring in proofs of $F$ and recursively check
their provability; we stop at any point during that process if all facts in
$F^\pi$ become provable (lines~\ref{alg:checkProvability:allProved},
\ref{alg:checkProvability:ref:allProved},
\ref{alg:checkProvability:repl:allProved},
and~\ref{alg:checkProvability:rules:allProved}). Lines
\ref{alg:checkProvability:ref}--\ref{alg:checkProvability:allProved} handle the
reflexivity rules \eqref{eq:eq-ref}: to check provability of ${c \sameAs c}$,
we recursively check the provability each fact containing $c$. Lines
\ref{alg:checkProvability:repl}--\ref{alg:checkProvability:repl:allProved}
handle replacement rules \eqref{eq:eq-repl1}--\eqref{eq:eq-repl3}: we
recursively check the provability of ${c \sameAs c}$ for each constant $c$
occurring in $F$. Finally,
lines~\ref{alg:checkProvability:rules:1}--\ref{alg:checkProvability:rules:allProved}
handle the rules in $\pi(\Pi)$: we consider each rule ${r \in \pi(\Pi)}$ whose
head matches $F$ and each substitution $\tau$ that matches the body of $r$ in
$I$, and we recursively check the provability of $\body{r}\tau$.

\interspacing

\noindent\textbf{Procedure} $\mathsf{BF^\sameAs()}$ computes the set ${D
\subseteq I}$ of `doubtful' facts. After initialising $D$ to $\pi(E^-)$
(lines~\ref{alg:BF-sameAs:update-E-D:1}--\ref{alg:BF-sameAs:update-E-D:2}), we
consider each fact ${F \in D}$
(lines~\ref{alg:BF-sameAs:extract-F}--\ref{alg:BF-sameAs:rules:processed}) and
determine whether some ${G \in F^\pi}$ is no longer provable
(line~\ref{alg:BF-sameAs:checkProvability}); if so, we add to $D$ all facts
that might be affected by the deletion of $G$. Lines
\ref{alg:BF-sameAs:repl}--\ref{alg:BF-sameAs:repl:add-G} handle rules
\eqref{eq:eq-repl1}--\eqref{eq:eq-repl3}; line \ref{alg:BF-sameAs:ref} handles
rules \eqref{eq:eq-ref}; and lines
\ref{alg:BF-sameAs:rules:1}--\ref{alg:BF-sameAs:rules:derive} handle
$\pi(\Pi)$: we identify each rule ${r \in \pi(\Pi)}$ where $F$ matches a body
atom of $r$, we evaluate the remaining body atoms of $r$ in $I$, and we add
$\head{r}\tau$ to $D$ for each $\tau$ such that ${\body{r}\tau \subseteq I}$.
Once $D$ is processed, $(\gamma, P\setminus\merged{P})$ reflects the changes to
$(\pi,I)$, which we exploit in Algorithm~\ref{alg:propagateChanges}.

\subsection{Optimisations}\label{sec:incremental:optimisations}

\noindent\textbf{Reflexivity.} Facts of the form ${F = c \sameAs c}$ can be
expensive for backward chaining: due to reflexivity rules~\eqref{eq:eq-ref}, in
lines~\ref{alg:checkProvability:ref}--\ref{alg:checkProvability:ref:allProved}
we may end up recursively proving each fact $G$ that mentions $c$. However, $F$
holds trivially if $E'$ contains a fact mentioning $c$, in which case we can
consider $F$ proven and avoid any recursion. This is implemented in
lines~\ref{alg:saturate:C:if-ref}--\ref{alg:saturate:C:if-ref:derive}.

\interspacing

\noindent\textbf{Avoiding Redundant Derivations.} Assume that $\Gamma$ contains
a rule ${y_1 \sameAs y_2 \leftarrow \triple{x}{\duri{R}}{y_1} \wedge
\triple{x}{\duri{R}}{y_2}}$, and consider a call to $\mathsf{saturate}()$ in
which facts $\triple{a}{R}{b}$ and $\triple{a}{R}{d}$ both end up in $P$.
Unless we are careful, in line~\ref{alg:saturate:F:derive} we might consider
substitution ${\tau_1 = \{x \mapsto a, \, y_1 \mapsto b, \, y_2 \mapsto d \}}$
twice: once when we match ${\triple{a}{R}{b}}$ to
${\triple{x}{\duri{R}}{y_1}}$, and once when we match ${\triple{a}{R}{d}}$ to
${\triple{x}{\duri{R}}{y_2}}$. Such redundant derivations can substantially
degrade performance.

To solve this problem, set $V$ keeps track of the processed subset of $P$:
after we extract a fact $F$ from $P$, in line~\ref{alg:saturate:add} we
transfer $F$ to $V$; moreover, in line~\ref{alg:saturate:P:rules:2} we evaluate
rule bodies in ${V \setminus \merged{P}}$ instead of ${P \setminus
\merged{P}}$. Now if ${\triple{a}{R}{b}}$ is processed before
${\triple{a}{R}{d}}$, at that point we have ${\triple{a}{R}{d} \not\in V}$, so
$\tau_1$ is not returned as a match in line~\ref{alg:saturate:P:rules:2}; the
situation when ${\triple{a}{R}{d}}$ is processed first is analogous. This,
however, does not eliminate all repetition: ${\tau_2 = \{ x \mapsto a, y_1
\mapsto b, y_2 \mapsto b \}}$ is still considered when ${\triple{a}{R}{b}}$ is
matched to either of the two body atoms in the rule. Therefore, we annotate
(see Section~\ref{sec:incremental:formalisation}) the body atoms of rules so
that, whenever $F$ is matched to some body atom $B_i$, no atom $B_j$ preceding
$B_i$ in the body of $r$ can be matched to $F$. In our example, $\tau_2$ is
thus considered only when ${\triple{a}{R}{b}}$ is matched to
${\triple{x}{\duri{R}}{y_1}}$.

\BFeq avoids redundant derivations in similar vein: set $O$ tracks the
processed subset of $D$; in lines~\ref{alg:BF-sameAs:repl:iterate-G}
and~\ref{alg:BF-sameAs:rules:2} we match the relevant rules in ${I \setminus
O}$; and in line~\ref{alg:BF-sameAs:rules:processed} we add a fact to $O$ once
it has been processed.

\interspacing

\noindent\textbf{Disproved Facts.} For each ${F \in I}$ with ${F^\pi \cap J' =
\emptyset}$, no fact in $F^\pi$ participates in a proof of any fact in $J'$.
Thus, in line~\ref{alg:BF-sameAs:disproved} we collect all such facts in a set
$S$ of \emph{disproved} facts, and in
lines~\ref{alg:checkProvability:ref:iterate-G},
\ref{alg:checkProvability:repl}, and~\ref{alg:checkProvability:rules:2} we
exclude $S$ from backward chaining.

\interspacing

\noindent\textbf{Singletons.} If we encounter ${F = c \sameAs c}$ in
line~\ref{alg:BF-sameAs:repl} or~\ref{alg:checkProvability:repl} where $c$
represents only itself (i.e., ${|c^\pi| = 1}$), then we know that no fact in
$F^\pi$ can derive a new fact using rules
\eqref{eq:eq-repl1}--\eqref{eq:eq-repl3}, and so we can avoid considering rules
\eqref{eq:eq-repl1}--\eqref{eq:eq-repl3}.

\subsection{Formalisation}\label{sec:incremental:formalisation}

\begin{figure*}[p]
\begin{minipage}[t]{0.47\textwidth}
\begin{tabular}[t]{@{}c@{\ \ }l@{}}
    \hline
    \multicolumn{2}{c}{\textbf{Input Variables}} \\
    \hline
    $E$             & : the explicit facts \\
    $\Pi$           & : the datalog program \\
    $(\pi,I)$       & : the r-materialisation of $E$ w.r.t.\ $\Pi$ \\
    $E^-$           & : the facts to delete from $E$ \\[2ex]
    \multicolumn{2}{c}{\textbf{Global Temporary Variables}} \\
    \hline
    $D$             & : the consequences of $E^-$ that might require deletion \\
    $O$             & : the processed subset of $D$ \\
    $C$             & : the facts whose provability must be checked \\
    $\gamma$        & : the mapping recording the changes needed to $\pi$ \\
    $P$             & : the proved facts \\
    $\merged{P}$    & : the proved rewritten facts  \\
    $Y$             & : the proved facts not in $C^\pi$ \\
    $V$             & : the processed subset of $P$ \\
    $S$             & : the set of disproved facts \\[-1.4ex]
\end{tabular}
\begin{algorithm}[H]
\caption{$\mathsf{B/F}^\sameAs()$}\label{alg:BF-sameAs}
\begin{algorithmic}[1]
    \similar\State $C \defeq D \defeq P \defeq \merged{P} \defeq Y \defeq O \defeq S \defeq V \defeq \emptyset$             \label{alg:BF-sameAs:initialise}
    \new\State initialise $\gamma$ as identity and $\Gamma \defeq \Pi$
    \For{\textbf{each} $F \in E^-$}                                                                                         \label{alg:BF-sameAs:update-E-D:1}
        \State \textbf{if} $\delete{E}{F}$ \textbf{then} $\add{D}{\pi(F)}$                                                  \label{alg:BF-sameAs:update-E-D:2}
    \EndFor
    \While{$(F \defeq \nnext{D}) \neq \varepsilon$}                                                                         \label{alg:BF-sameAs:extract-F}
        \State $\mathsf{checkProvability}(F)$                                                                               \label{alg:BF-sameAs:checkProvability}
        \similar\For{\textbf{each} $G \in C$ s.t.\ $\mathsf{allDisproved}(G)$}
            $\add{S}{G}$                                                                                                    \label{alg:BF-sameAs:disproved}
        \EndFor
        \similar\If{not $\mathsf{allProved}(F)$}                                                                            \label{alg:BF-sameAs:notAllProved}
            \new\If{$F = c \sameAs c$ and $|c^\pi| > 1$}                                                                    \label{alg:BF-sameAs:repl}
                \new\For{\textbf{each} $G \in I \setminus O$ with $c \in \voc{G}$}                                          \label{alg:BF-sameAs:repl:iterate-G}
                    \new\State $\add{D}{G}$                                                                                 \label{alg:BF-sameAs:repl:add-G}
                \EndFor
            \EndIf
            \new\For{\textbf{each} $c \in \voc{F}$}
                $\add{D}{c \sameAs c}$                                                                                      \label{alg:BF-sameAs:ref}
            \EndFor
            \For{\textbf{each} $\langle r,Q,\sigma \rangle \in \matchBody{\pi(\Pi)}{F}$}                                    \label{alg:BF-sameAs:rules:1}
                \For{\textbf{each} $\tau \in \evaluate{[I \setminus O]}{Q}{\{ F \}}{\sigma}$}                               \label{alg:BF-sameAs:rules:2}
                    \State $\add{D}{\head{r}\tau}$                                                                          \label{alg:BF-sameAs:rules:derive}
                \EndFor
            \EndFor
            \State $\add{O}{F}$                                                                                             \label{alg:BF-sameAs:rules:processed}
        \EndIf
    \EndWhile
    \similar\State $\mathsf{propagateChanges}()$                                                                            \label{alg:BF-sameAs:propagateChanges}
    \algstore{B-F-delete}
\end{algorithmic}
\end{algorithm}
\vspace{-0.9cm}
\begin{algorithm}[H]\newalgorithm
\caption{$\mathsf{propagateChanges}()$}\label{alg:propagateChanges}
\begin{algorithmic}[1]
    \algrestore{B-F-delete}
    \For{\textbf{each} $c \sameAs c \in C$ and each $d$ with $\pi(d) = c$}                                                  \label{alg:propagateChanges:pi-1}
        \State $\pi(d) \defeq \gamma(d)$                                                                                    \label{alg:propagateChanges:pi-2}
    \EndFor
    \For{\textbf{each} $F \in D \setminus (P \setminus \merged{P})$}
        $\delete{I}{F}$                                                                                                     \label{alg:propagateChanges:delete-unproved}
    \EndFor
    \For{\textbf{each} $F \in P \setminus \merged{P}$}
        $\add{I}{\pi(F)}$                                                                                                   \label{alg:propagateChanges:add-proved}
    \EndFor
    \algstore{B-F-delete}
\end{algorithmic}
\end{algorithm}
\vspace{-0.9cm}
\begin{algorithm}[H]\newalgorithm
\caption{Auxiliary functions}\label{alg:auxiliary-functions}
\begin{algorithmic}
    \Statex $\mathsf{allProved}(F)$:
    \Statex \qquad $\true$ iff $F \not\in S$ and $\gamma(F^\pi) \subseteq (P \setminus \merged{P})$
    \vspace{1ex}
    \Statex $\mathsf{allDisproved}(F)$:
    \Statex \qquad $\true$ iff $\gamma(F^\pi) \cap (P \setminus \merged{P}) = \emptyset$
\end{algorithmic}
\end{algorithm}
\end{minipage}
\;\;
\begin{minipage}[t]{0.49\textwidth}
\vspace{-0.45cm}
\begin{algorithm}[H]
\caption{$\mathsf{checkProvability}(F)$}\label{alg:checkProvability}
\begin{algorithmic}[1]
    \algrestore{B-F-delete}
    \If{not $\add{C}{F}$}
        \Return                                                                                                             \label{alg:checkProvability:add-F}
    \EndIf
    \State $\mathsf{saturate}()$                                                                                            \label{alg:checkProvability:saturate}
    \similar\If{$\mathsf{allProved}(F)$}                                                                                    \label{alg:checkProvability:allProved}
        \similar\Return
    \similar\EndIf
    \new\If{$F = c \sameAs c$}                                                                                              \label{alg:checkProvability:ref}
        \new\For{\textbf{each} $G \in I \setminus S$ with $c \in \voc{G}$}                                                  \label{alg:checkProvability:ref:iterate-G}
            \new\State $\mathsf{checkProvability}(G)$                                                                       \label{alg:checkProvability:ref:recurse}
            \new\If{$\mathsf{allProved}(F)$}                                                                                \label{alg:checkProvability:ref:allProved}
                \Return
            \EndIf
        \EndFor
    \EndIf
    \new\For{\textbf{each} $c \in \voc{F}$ with $c \sameAs c \not\in S$ and $|c^\pi| > 1$}                                  \label{alg:checkProvability:repl}
        \new\State $\mathsf{checkProvability}(c \sameAs c)$                                                                 \label{alg:checkProvability:repl:recurse}
        \new\If{$\mathsf{allProved}(F)$}                                                                                    \label{alg:checkProvability:repl:allProved}
            \Return
        \EndIf
    \EndFor
    \For{\textbf{each} $\langle r,Q,\sigma \rangle \in \matchHead{\pi(\Pi)}{F}$}                                            \label{alg:checkProvability:rules:1}
        \For{\textbf{each} $\tau \in \evaluate{[I \setminus S]}{Q}{\emptyset}{\sigma}$ and $G \in \body{r}\tau$}            \label{alg:checkProvability:rules:2}
            \State $\mathsf{checkProvability}(G)$                                                                           \label{alg:checkProvability:rules:recurse}
            \If{$\mathsf{allProved}(F)$}                                                                                    \label{alg:checkProvability:rules:allProved}
                \Return
            \EndIf
        \EndFor
    \EndFor
    \algstore{B-F-delete}
\end{algorithmic}
\end{algorithm}
\vspace{-0.9cm}
\begin{algorithm}[H]
\caption{$\mathsf{saturate}()$}\label{alg:saturate}
\begin{algorithmic}[1]
    \algrestore{B-F-delete}
    \While{$(F \defeq \nnext{C}) \neq \varepsilon$}                                                                         \label{alg:saturate:C}
        \new\If{$F = c \sameAs c$}                                                                                          \label{alg:saturate:C:if-ref}
            \new\For{\textbf{each} $d \in \voc{E}$ with $\pi(d) = c$}                                                       \label{alg:saturate:C:if-ref:check-E}
                \new\State $\add{P}{\gamma(d) \sameAs \gamma(d)}$                                                           \label{alg:saturate:C:if-ref:derive}
            \EndFor
        \EndIf
        \similar\For{\textbf{each} $G \in F^\pi \cap(E \cup Y)$}                                                            \label{alg:saturate:C:expand-F-derive}
             $\add{P}{\gamma(G)}$
        \EndFor
    \EndWhile
    \While{$(F \defeq \nnext{P}) \neq \varepsilon$}                                                                         \label{alg:saturate:F}
        \similar\If{$F \in P \setminus (\merged{P} \cup V)$ and $\add{V}{F}$}                                               \label{alg:saturate:add}
            \new\State $G \defeq \gamma(F)$
            \new\If{$F \neq G$}
                $\add{\merged{P}}{F}$ and $\add{P}{G}$                                                                      \label{alg:saturate:RewF}
            \new\ElsIf{$F = a \sameAs b$ and $a \neq b$}
                $\mathsf{rewrite}(a,b)$                                                                                     \label{alg:saturate:RewC}
            \new\Else
                \similar\For{\textbf{each} $c \in \voc{G}$}
                    $\mathsf{prove}(c \sameAs c)$                                                                           \label{alg:saturate:P:ref}
                \EndFor
                \For{\textbf{each} $\langle r,Q,\sigma \rangle \in \matchBody{\Gamma}{G}$}                                  \label{alg:saturate:P:rules:1}
                    \similar\For{\textbf{each} $\tau \in \evaluate{[V \setminus \merged{P}]}{Q}{\{ G \}}{\sigma}$}          \label{alg:saturate:P:rules:2}
                        \similar\State $\mathsf{prove}(\head{r}\tau)$                                                       \label{alg:saturate:F:derive}
                    \EndFor
                \EndFor
            \EndIf
        \EndIf
    \EndWhile
    \algstore{B-F-delete}
\end{algorithmic}
\end{algorithm}
\vspace{-0.9cm}
\begin{algorithm}[H]\newalgorithm
\caption{$\mathsf{rewrite}(a,b)$}\label{alg:rewrite}
\begin{algorithmic}[1]
    \algrestore{B-F-delete}
    \State $c \defeq \min \{ a,b \}$ \quad $d \defeq \max \{ a,b \}$
    \State $\mergeInto{\gamma}{d}{c}$                                                                                       \label{alg:rewrite:merge}
    \For{\textbf{each} $F \in P \setminus \merged{P}$ with $d \in \voc{F}$}
        \State $\add{\merged{P}}{F}$ and $\add{P}{\gamma(F)}$                                                               \label{alg:rewrite:normF}
    \EndFor
    \For{\textbf{each} $r \in \Gamma$ with $r \neq \gamma(r)$}
        \State replace $r$ in $\Gamma$ with $r' \defeq \gamma(r)$                                                           \label{alg:rewrite:normR}
        \For{\textbf{each} $\tau \in \evaluate{[V \setminus \merged{P}]}{\body{r'}}{\emptyset}{\emptyset}$}                 \label{alg:rewrite:reapply}
            \State $\mathsf{prove}(\head{r'}\tau)$                                                                          \label{alg:rewrite:prove}
        \EndFor
    \EndFor
    \algstore{B-F-delete}
\end{algorithmic}
\end{algorithm}
\vspace{-0.9cm}
\begin{algorithm}[H]\newalgorithm
\caption{$\mathsf{prove}(F)$}\label{alg:prove}
\begin{algorithmic}[1]
    \algrestore{B-F-delete}
    \State\textbf{if} $\pi(F) \in C$ \textbf{ then } $\add{P}{F}$ \textbf{ else } $\add{Y}{F}$                              \label{alg:prove:add}
\end{algorithmic}
\end{algorithm}
\end{minipage}
\bigskip
\begin{restatable}{theorem}{thmcorrectness}\label{thm:correctness}
    Let $(\pi,I)$ be the r-materialisation of a dataset $E$ w.r.t.\ a program
    $\Pi$, and let $E^-$ be a dataset.
    \begin{enumerate}
        \item Algorithm \ref{alg:BF-sameAs} terminates, at which point
        $(\pi,I)$ contains the r-materialisation of ${E \setminus E^-}$ w.r.t.\
        $\Pi$.

        \item Each combination of a rule $r$ and a substitution $\tau$ is
        considered at most once in line~\ref{alg:saturate:F:derive} or
        line~\ref{alg:rewrite:prove}, but not both.

        \item Each combination of a rule $r$ and a substitution $\tau$ is
        considered at most once in line~\ref{alg:BF-sameAs:rules:derive}.
    \end{enumerate}
\end{restatable}
\end{figure*}

We borrow the notation by \citeA{mnph15incremental-BF} to formalise \BFeq. We
recapitulate some definitions, present the pseudo-code, and formally state the
algorithm's properties.

Given a dataset $X$ and a fact $F$, operation $\add{X}{F}$ adds $F$ to $X$, and
operation $\delete{X}{F}$ removes $F$ from $X$; both return $\true$ if $X$ was
changed. For iteration, operation $\nnext{X}$ returns the next fact from $X$,
or $\nofact$ if no such fact exists.

An \emph{annotated query} has the form ${Q = B_1^{\bowtie_1} \wedge \dots
\wedge B_k^{\bowtie_k}}$, where each $B_i$ is an atom and \emph{annotation}
$\bowtie_i$ is either empty or equal to $\neq$. Given datasets $X$ and $Y$ and
a substitution $\sigma$, operation $\evaluate{X}{Q}{Y}{\sigma}$ returns a set
containing each smallest substitution $\tau$ such that ${\sigma \subseteq
\tau}$ and, for ${1 \leq i \leq k}$, (i)~${B_i\tau \in X}$ if $\bowtie_i$ is
empty or (ii)~${B_i\tau \in X \setminus Y}$ if $\bowtie_i$ is $\neq$. We often
write $[Z \setminus W]$ instead of $X$, meaning that $Q$ is evaluated in the
difference of sets $Z$ and $W$.

Given a fact $F$, operation $\matchHead{\Pi}{F}$ returns all tuples ${\langle
r,Q,\sigma \rangle}$ with ${r \in \Pi}$ a rule of the form \eqref{eq:rule},
$\sigma$ a substitution such that ${H\sigma = F}$, and ${Q = B_1 \wedge \dots
\wedge B_n}$. Moreover, operation $\matchBody{\Pi}{F}$ returns all tuples
${\langle r,Q,\sigma \rangle}$ with ${r \in \Pi}$ a rule of the form
\eqref{eq:rule}, $\sigma$ a substitution such that ${B_i\sigma = F}$ for some
${1 \leq i \leq n}$, and $Q$ is defined as
\begin{align}
    Q = B_1^{\neq} \wedge \dots \wedge B_{i-1}^{\neq} \wedge B_{i+1} \wedge \dots \wedge B_n. \label{eq:seminaive-annotated-query}
\end{align}

Finally, given a mapping $\gamma$ of constants to constants, and constants $d$
and $c$, operation $\mergeInto{\gamma}{d}{c}$ modifies $\gamma$ so that
${\gamma(e) = c}$ holds for each constant $e$ with ${\gamma(e) = d}$.

\BFeq consists of Algorithms~\ref{alg:BF-sameAs}--\ref{alg:prove}.
Theorem~\ref{thm:correctness} shows that the algorithm is correct and that,
just like the semina{\"i}ve algorithm \cite{abiteboul95foundation}, it does not
repeat derivations; the proof is given in \ifdraft{the appendix}{a technical
report \cite{DBLP:journals/corr/MotikNPH15}}.

%% file: evaluation.tex
\section{Evaluation}\label{sec:evaluation}

\begin{table*}[ht!]
\begin{center}\scriptsize
\input{results}
\end{center}
\caption{Experimental results}\label{tab:results}
\end{table*}

We have implemented and evaluated the \BFeq algorithm in the open-source RDF
data management system RDFox. The system and the test data are all available
online.\footnote{\url{https://krr-nas.cs.ox.ac.uk/2015/IJCAI/RDFox/index.html}}

\interspacing

\noindent\textbf{Objectives.} Updates can be handled either incrementally or by
rematerialisation, and equality can be handled either by rewriting or by
axiomatisation, giving rise to four possible approaches to updates. Our first
objective was to compare all of them to determine their relative strengths and
weaknesses.

As $E^-$ increases in size, incremental update becomes harder, but
rematerialisation becomes easier. Thus, our second objective was to investigate
the relationship between the update size and the performance of the respective
approaches.

\interspacing

\noindent\textbf{Datasets.} Equality is often used in OWL ontologies on the
Semantic Web, so we based our evaluation on several well-known synthetic and
`real' RDF datasets.

Each dataset comprises an OWL ontology and a set of explicit facts $E$.
\emph{UOBM} \cite{DBLP:conf/esws/MaYQXPL06} extends LUBM
\cite{DBLP:journals/ws/GuoPH05}, and we used the data generated for 100
universities; we did not use LUBM because it does not use $\sameAs$.
\emph{Claros} contains information about cultural
artefacts.\footnote{\url{http://www.clarosnet.org/XDB/ASP/clarosHome/}}
\emph{DBpedia} consists of structured information extracted from
Wikipedia.\footnote{\url{http://dbpedia.org/}} \emph{UniProt} is a knowledge
base about protein sequences;\footnote{\url{http://www.uniprot.org}} we
selected a subset of the original (very large) set of facts. Finally,
\emph{OpenCyc} is an extensive, manually curated upper
ontology.\footnote{\url{http://www.cyc.com/platform/opencyc}}

Following \citeA{DBLP:conf/www/ZhouGHWB13}, we converted the ontologies into
\emph{lower} (\textbf{L}) and \emph{upper bound} (\textbf{U}) programs: the
former is the OWL 2 RL subset of the ontology transformed into datalog as
described by \citeA{GHVD03}, and the latter captures all consequences of the
ontology using an unsound approximation. Upper bound programs are interesting
as they tend to be `hard'. We also manually extended the lower bound
(\textbf{LE}) of Claros with `hard' rules (e.g., we defined related documents
as pairs of documents that refer to the same topic).

\interspacing

\noindent\textbf{Update Sets.} For each dataset, we randomly selected several
subsets $E^-$ of $E$. We considered small updates of 100 and 5k facts on all
datasets. Moreover, for each dataset we identified the `equilibrium' point $n$
at which \BFeq and Remat${^\sameAs}$ take roughly the same time. If $n$ was
large, we generated subsets $E^-$ with sizes equal to 25\%, 50\%, 75\%, and
100\% of $n$; otherwise, we divided $n$ in an ad hoc way.

\interspacing

\noindent\textbf{Test Setting.} We used a Dell server with two 2.60GHz Intel
Xeon E5-2670 CPUs and 256 GB of RAM running Fedora release 20, kernel version
3.17.7-200.fc20.x86\_64.

\interspacing

\noindent\textbf{Test Results.} Table~\ref{tab:results} summarises our test
results. For each dataset, we show the numbers of explicit facts ($|E|$) and
rules ($|\Pi|$), the number of facts in the initial r-materialisation
($|I^\sameAs|$), and the time ($T^\sameAs$) and the number of derivations
($D^\sameAs$) used to compute it via rewriting; moreover, we show the latter
three numbers for the initial materialisation computed using axiomatised
equality ($|I^A|$, $T^A$, and $D^A$). For each set $E^-$, we show the numbers
${\Delta |I^\sameAs|}$ and ${\Delta |I^A|}$ of deleted facts with rewriting and
axiomatisation, respectively, as well as the times (T) and the number of
derivations (D) for each of the four update approaches. All times are in
seconds. We could not complete all axiomatisation tests with Claros-LE as each
run took about two hours. Due to the upper bound transformation, the
r-materialisation of UOBM-100-U contains a constant $c$ with ${|c^\pi| =
3930}$; thus, when $\sameAs$ is axiomatised, deriving just all equalities
involving $c^\pi$ requires $3930^3 = 60$ billion derivations, which causes the
initial materialisation to last longer than four hours. The number of
derivations $D$ in \BFeq is the sum of the number of times a fact is determined
as `doubtful' (lines \ref{alg:BF-sameAs:repl:add-G}, \ref{alg:BF-sameAs:ref},
and \ref{alg:BF-sameAs:rules:derive}), checked in backward chaining (lines
\ref{alg:checkProvability:ref:recurse},
\ref{alg:checkProvability:repl:recurse}, and
\ref{alg:checkProvability:rules:recurse}), or derived in forward chaining (line
\ref{alg:prove:add}); we use this number to estimate reasoning difficulty
independently from implementation details.

\interspacing

\noindent\textbf{Discussion.} For updates of 100 facts, \BFeq outperforms all
other approaches, often by orders of magnitude, and in most cases it does so
even for much larger updates.

Even when ${|I^A| - |I^\sameAs|}$ is `small' (i.e., when not many equalities
are derived), \BFeq outperforms \BFA. This seems to be mainly because \BFA
ascribes no special meaning to $\PsameAs$ and so it does not use the
optimisation from
lines~\ref{alg:saturate:C:if-ref}--\ref{alg:saturate:C:if-ref:derive}; thus,
when trying to prove ${c \sameAs c}$, \BFA performs backward chaining via rules
\eqref{eq:eq-ref} and so it potentially examines each fact containing $c$. On
Claros-L, although $|I^A|$ and $|I^\sameAs|$ are of similar sizes, $I^A$
contains one constant $c$ with ${|c^\pi| = 306}$, which gives rise to $306^3$
derivations; this explains the difference in the performance of \BFeq and \BFA.

\Remateq outperforms \BFeq in cases similar to those described by
\citeA{mnph15incremental-BF}. For example, in UOBM, relation
$\duri{hasSameHomeTownWith}$ is symmetric and transitive, which creates cliques
of connected constants; \BF always recomputes each changed clique, thus
repeating most of the `hard' work. Equality connects constants in cliques,
which poses similar problems for \BFeq. For example, due to the constant $c$
with ${|c^\pi| = 3930}$, deleting 5k facts in UOBM-100-U results in only 961k
(about 1.2\% of $|I^\sameAs|$) facts being added to set $C$ in
line~\ref{alg:checkProvability:add-F}, but these facts contribute to 73\% of
the derivations from the initial r-materialisation; thus, \BFeq repeats in
Algorithm~\ref{alg:saturate} a substantial portion of the initial work.

On OpenCyc-L, \Remateq already outperforms \BFeq on updates of 1k triples,
which was surprising since the former makes more derivations than the latter.
Our investigation revealed that OpenCyc-L contains about 200 rules of the form
${\triple{x}{\rdfType}{y} \leftarrow \triple{x}{R_i}{y}}$ that never fire
during forward chaining; however, to check provability of
$\triple{a}{\rdfType}{C}$, Algorithm~\ref{alg:checkProvability} considers in
line~\ref{alg:checkProvability:rules:1} each time each of the 200 rules. After
removing all such `idle' rules manually, \BFeq and \Remateq could update 1k
tuples in roughly the same time. Further analysis revealed that the slowdown in
\BFeq occurs mainly in line~\ref{alg:saturate:C:expand-F-derive}: the condition
is checked for 13.3M facts $F$, and these give rise to 139M facts in $F^\pi$,
each requiring an index lookup; the latter number is similar to the number of
derivations in rematerialisation, which explains the slowdown. We believe one
can check this condition more efficiently by using additional book-keeping.

%% file: results.tex
\begin{results}{UOBM-100-L}{24.5M}{210}{46.4M}{69}{79.3M}{46.7M}{122}{361M}
    \result{100}
        {146}
            {0.6}{0.7k}
            {45.1}{79.3M}
        {146}
            {4.6}{32.8k}
            {94.6}{361M}

    \result{5k}
        {7.8k}
            {1.2}{45.8k}
            {42.5}{79.3M}
        {7.9k}
            {7.1}{805k}
            {93.1}{361M}

    \result{1.3M}
        {1.9M}
            {18.2}{8.7M}
            {39.2}{75.4M}
        {2.0M}
            {38.0}{98.0M}
            {89.5}{361M}

    \result{2.5M}
        {3.9M}
            {29.9}{15.8M}
            {41.7}{71.5M}
        {4.0M}
            {54.5}{151M}
            {83.7}{345M}

    \result{3.8M}
        {5.8M}
            {31.8}{22.3M}
            {37.4}{67.7M}
        {5.9M}
            {70.9}{188M}
            {79.3}{329M}

    \result{5M}
        {7.7M}
            {41.2}{28.4M}
            {36.2}{63.8M}
        {7.9M}
            {73.8}{218M}
            {73.2}{314M}
\end{results}
\begin{results}{UOBM-100-U}{24.5M}{279}{78.8M}{225}{719M}{---}{---}{---}

    \result{100}
        {197}
            {3.5}{21.5k}
            {209}{719M}
        {---}
            {\phantom{nnii}---}{\phantom{nnii}---}
            {\phantom{nnii}---}{\phantom{nnii}---}

    \result{1k}
        {1.8k}
            {277}{581M}
            {219}{719M}
        {---}
            {---}{---}
            {---}{---}

    \result{2.5k}
        {4.3k}
            {338}{584M}
            {209}{719M}
        {---}
            {---}{---}
            {---}{---}

    \result{5k}
        {8.5k}
            {345}{584M}
            {214}{719M}
        {---}
            {---}{---}
            {---}{---}

\end{results}
\newresultsline
\begin{results}{Claros-L}{18.8M}{1.3k}{79.5M}{83}{129M}{102M}{3.9k}{11.0G}
    \result{100}
        {209}
            {8.3}{797k}
            {77.4}{135M}
        {819}
            {2476}{15.8G}
            {3174}{11.0G}

    \result{5k}
        {11.2k}
            {9.1}{895k}
            {77.0}{135M}
        {18.6k}
            {2609}{15.8G}
            {3166}{11.0G}

    \result{750k}
        {1.7M}
            {29.5}{14.5M}
            {80.9}{131M}
        {4.0M}
            {2816}{17.1G}
            {2690}{9.5G}

    \result{1.5M}
        {3.5M}
            {46.1}{26.5M}
            {81.5}{127M}
        {10.1M}
            {2757}{17.4G}
            {1933}{7.3G}

    \result{2.3M}
        {5.3M}
            {63.9}{38.4M}
            {77.7}{123M}
        {15.3M}
            {3092}{18.3G}
            {1389}{5.5G}

    \result{3M}
        {7.2M}
            {78.4}{48.8M}
            {72.4}{119M}
        {19.4M}
            {3170}{18.6G}
            {1075}{4.4G}
\end{results}
\begin{results}{Claros-LE}{18.8M}{1.3k}{539M}{4514}{12.6G}{562M}{9048}{26.3G}
    \result{100}
        {522}
            {16.1}{617k}
            {4397}{12.6G}
        {1132}
            {5703}{25.8G}
            {8693}{26.3G}

    \result{2.5k}
        {179k}
            {31.6}{9.9M}
            {4430}{12.6G}
        {---}
            {---}{---}
            {---}{---}

    \result{5k}
        {427k}
            {39.4}{10.7M}
            {4392}{12.6G}
        {435k}
            {5845}{25.8G}
            {9383}{26.3G}

    \result{7.5k}
        {609k}
            {44.8}{11.6M}
            {4713}{12.6G}
        {---}
            {---}{---}
            {---}{---}

    \result{10k}
        {781k}
            {4300}{12.4G}
            {4627}{12.6G}
        {---}
            {---}{---}
            {---}{---}

\end{results}

\newresultsline

\begin{results}{DBpedia-L}{113M}{3.4k}{136M}{49.3}{36.6M}{139M}{641}{895M}
    \result{100}
        {105}
            {0.3}{91}
            {47.5}{36.6M}
        {105}
            {8.9}{1.7M}
            {251}{895M}

    \result{5k}
        {5.0k}
            {0.4}{24.4k}
            {64.6}{36.6M}
        {5.3k}
            {20.3}{5.7M}
            {256}{895M}

    \result{1.8M}
        {1.8M}
            {29.4}{2.1M}
            {48.7}{36.3M}
        {2.0M}
            {50.0}{72.2M}
            {239}{895M}

    \result{3.5M}
        {3.6M}
            {38.9}{3.6M}
            {49.0}{35.9M}
        {3.9M}
            {85.5}{116M}
            {237}{881M}

    \result{5.3M}
        {5.3M}
            {52.2}{4.9M}
            {54.3}{35.5M}
        {5.9M}
            {89.8}{152M}
            {232}{866M}

    \result{7M}
        {7.1M}
            {63.1}{6.2M}
            {50.7}{35.1M}
        {7.8M}
            {103}{184M}
            {227}{852M}
\end{results}
\begin{results}{UniProt-L}{123M}{451}{179M}{118}{183M}{229M}{527}{1.6G}
    \result{100}
	{125}
	   {2.5}{892}
	   {235}{238M}
	{125}
	   {14.3}{6.0k}
	   {490}{1.6G}

    \result{5k}
	{6.1k}
	   {3.4}{35k}
	   {221}{238M}
	{6.1k}
	   {17.5}{271k}
	   {482}{1.6G}

    \result{4.5M}
	{5.7M}
	   {84.0}{24.8M}
	   {204}{232M}
	{5.7M}
	   {125}{190M}
	   {475}{1.5G}

    \result{9M}
	{11.5M}
	    {137}{46.7M}
	    {216}{225M}
	{11.5M}
	    {192}{344M}
	    {478}{1.5M}

    \result{13.5M}
	{17.4M}
	    {209}{67.1M}
	    {220}{218M}
	{17.4M}
	    {315}{483M}
	    {473}{1.4G}

    \result{18M}
	{23.4M}
	    {220}{86.5M}
	    {217}{210M}
	{23.4M}
	    {371}{613M}
	    {481}{1.4G}

\end{results}
\newresultsline
\begin{results}{OpenCyc-L}{2.4M}{261k}{141M}{164}{280M}{1.2G}{3.5k}{12.9G}
    \result{100}
        {5.4k}
            {15.5}{405k}
            {220}{280M}
        {50.0k}
            {472}{8.5M}
            {3296}{12.9G}

    \result{1k}
        {53.1k}
            {1062}{69.5M}
            {222}{280M}
        {5.1M}
            {5537}{2.0G}
            {3479}{12.9G}

    \result{2.5k}
        {130k}
            {1078}{69.8M}
            {178}{279M}
        {5.8M}
            {5339}{2.1G}
            {3621}{12.8G}

    \result{5k}
        {261k}
            {1123}{70.4M}
            {177}{279M}
        {7.2M}
            {5475}{2.1G}
            {3334}{12.8G}
\end{results}

%% file: conclusion.tex
\section{Conclusion}\label{sec:conclusion}

This paper describes what we believe to be the first approach to incremental
maintenance of datalog materialisation when the latter is computed using
rewriting---a common optimisation used when programs contain equality. Our
algorithm proved to be very effective, particularly on small updates.

In our future work, we shall aim to address the issues we identified in
Section~\ref{sec:evaluation}. For example, to optimise the check in
line~\ref{alg:saturate:C:expand-F-derive}, we shall investigate ways of keeping
track of how explicit facts are merged so that we can implement the test by
iterating over the appropriate subset of $E$ rather than over $F^\pi$.
Moreover, we believe we can considerably improve the efficiency of both the
initial materialisation and the incremental updates by using specialised
algorithms for rules that produce large cliques; hence, we shall identify
common classes of `hard' rules and then develop such specialised algorithms.

%% file: appendix.tex
\appendix
\section{Proof of Theorem \ref{thm:correctness}}

Let $\Pi$ be a program (that ascribes no special meaning to $\sameAs$), and let
$E$ be a dataset. A \emph{derivation tree} for a fact $F$ from $E$ w.r.t.\
$\Pi$ is a finite tree $T$ in which each node $t$ is labelled with a fact
$\dtF{t}$, and each nonleaf node $t$ is labelled with a rule ${\dtR{t} \in
\Pi}$ and a substitution $\dtS{t}$ such that the following holds:
\begin{enumerate}
    \item[D1.]\label{dt:root} ${\dtF{\epsilon} = F}$ holds for the root
    $\epsilon$ of $T$;

    \item[D2.]\label{dt:leaf} ${\dtF{t} \in E}$ holds for each leaf node $t$ of
    $T$; and

    \item[D3.]\label{dt:nonleaf} ${\head{\dtR{t}}\dtS{t} = \dtF{t}}$ and
    ${\body{\dtR{t}}\dtS{t} = \{ \dtF{t_1}, \dots, \dtF{t_n} \}}$ hold for each
    nonleaf node $t$ of $T$ with children ${t_1, \dots, t_n}$.
\end{enumerate}
The \emph{materialisation} $\Pi^\infty(E)$ of $E$ w.r.t.\ $\Pi$ is the smallest
set containing each fact that has a derivation tree from $E$ w.r.t.\ $\Pi$;
this definition of $\Pi^\infty(E)$ is equivalent to the one in
Section~\ref{sec:preliminaries}. The \emph{height} of a derivation tree is the
length of its longest branch; moreover, the \emph{height} of a fact ${F \in
\Pi^\infty(E)}$ w.r.t.\ $E$ and $\Pi$ is the minimum height of a derivation
tree for $F$ from $E$ w.r.t.\ $\Pi$.

In the rest of this paper, we make the following assumption ($\ast$): no
derivation tree contains a node $t$ where $\dtR{t}$ is \eqref{eq:eq-repl1} and
${\dtS{t}(x_1) = \dtS{t}(x_1')}$, or $\dtR{t}$ is \eqref{eq:eq-repl2} and
${\dtS{t}(x_2) = \dtS{t}(x_2')}$, or $\dtR{t}$ is \eqref{eq:eq-repl3} and
${\dtS{t}(x_3) = \dtS{t}(x_3')}$ . This is w.l.o.g.\ because, for each such
$t$, we have ${\dtF{t} = \dtF{t_1}}$ for $t_1$ the first child of $t$; hence,
we can always remove such $t$ from the derivation tree.

\smallskip

Next, we recapitulate Theorem \ref{thm:correctness} and present its proof,
which we split into several claims.

\thmcorrectness*

In the rest of this section, we fix a datalog program $\Pi$ and datasets $E$
and $E^-$. Let $(\pi,I)$ be the r-materialisation of $E$ w.r.t.\ $\Pi$; let ${J
\defeq (\Pi \cup \PsameAs)^\infty(E)}$; let ${E' \defeq E \setminus E^-}$; let
$(\pi',I')$ be the r-materialisation of $E'$ w.r.t.\ $\Pi$; and let ${J' \defeq
(\Pi \cup \PsameAs)^\infty(E')}$. By the monotonicity of datalog, we clearly
have ${J' \subseteq J}$.

We next show that Algorithm~\ref{alg:saturate} essentially captures the
r-materialisation algorithm by \citeA{mnph15owl-sameAs-rewriting}.

\begin{claim}\label{claim:saturate}
    Let $P$ and $\merged{P}$ be as obtained after a call to
    Algorithm~\ref{alg:saturate} in line~\ref{alg:checkProvability:saturate},
    let ${K \defeq \{ d \sameAs d \mid d \in \voc{E} \}}$, and let $L$ be the
    set containing precisely each fact $F$ that has a derivation $T$ from ${K
    \cup E'}$ w.r.t.\ ${\Pi \cup \PsameAs}$ in which ${\dtF{t} \in C^\pi}$
    holds for each node $t$ of $T$. Then, the following properties hold:
    \begin{enumerate}
        \item ${\gamma(c) = \min \eqclass{L}{c}}$
        for each constant $c$;

        \item ${P \setminus \merged{P} = \gamma(L)}$; and

        \item each combination of a rule $r$ and a substitution $\tau$ is
        considered at most once in line~\ref{alg:saturate:F:derive} or
        line~\ref{alg:rewrite:prove}, but not both.
    \end{enumerate}
\end{claim}

\begin{proof}[Proof (Sketch)]
Algorithm~\ref{alg:saturate} is a variant of the r-materialisation algorithm by
\citeA{mnph15owl-sameAs-rewriting}, so properties 1--3 hold by a
straightforward modification of the correctness proof of that algorithm. This
proof is quite lengthy so, for the sake of brevity, we just summarise the
differences.
\begin{itemize}
    \item Lines \ref{alg:saturate:C:if-ref}--\ref{alg:saturate:C:if-ref:derive}
    ensure ${\gamma(C^\pi \cap K) \subseteq P \setminus \merged{P}}$, and line
    \ref{alg:saturate:C:expand-F-derive} ensures ${\gamma(C^\pi \cap E')
    \subseteq P \setminus \merged{P}}$; hence, ${C^\pi \cap (K \cup E')}$ plays
    the same role that explicit facts play in the algorithm by
    \citeA{mnph15owl-sameAs-rewriting}.

    \item Let $F$ be an arbitrary fact considered in line~\ref{alg:saturate:F}.
    To ensure property 4 of Claim~\ref{claim:saturate}, the algorithm by
    \citeA{mnph15owl-sameAs-rewriting} uses slightly different annotated
    queries to apply the rules in lines
    \ref{alg:saturate:P:rules:1}--\ref{alg:saturate:P:rules:2} only to facts
    extracted before $F$. In contrast, Algorithm~\ref{alg:prove} keeps track of
    previously processed facts in set $V$, but this has exactly the same effect.

    \item All derivations of a fact in line \ref{alg:saturate:P:ref},
    \ref{alg:saturate:F:derive}, or \ref{alg:rewrite:prove}, are handled by
    Algorithm~\ref{alg:prove}, which, for each $F$, checks whether ${\pi(F) \in
    C}$; this is equivalent to checking ${F \in C^\pi}$. If the latter holds,
    then $F$ is added to $P$, and otherwise $F$ is added to $Y$. If in a
    subsequent invocation of Algorithm~\ref{alg:saturate} set $C$ is extended
    such that ${\pi(F) \in C}$ suddenly holds, then $\gamma(F)$ is added to $P$
    in line~\ref{alg:saturate:C:expand-F-derive}. This, however, does not
    change the algorithm in any substantial way. \qedhere
\end{itemize}
\end{proof}

The following claim follows immediately from the definitions in
Algorithm~\ref{alg:auxiliary-functions}.

\begin{claim}\label{claim:allDisProved}
    The following properties hold for an arbitrary fact $F$ normal w.r.t.\
    $\pi$:
    \begin{enumerate}
        \item ${\mathsf{allProved}(F) = \true}$ if and only if ${F \not\in }S$
        and ${F^\pi \subseteq (P \setminus \merged{P})^\gamma}$; and

        \item ${\mathsf{allDisproved}(F) = \true}$ if and only if ${F^\pi \cap
        (P \setminus \merged{P})^\gamma = \emptyset}$.
    \end{enumerate}
\end{claim}

We next show that sets $C$, $P$, $\merged{P}$, $S$, and $\gamma$ always satisfy
an important property.

\begin{claim}\label{claim:checkProvability}
    Assume that Algorithm~\ref{alg:checkProvability} is applied to some fact
    $F$, mapping $\gamma$, and sets $S$, $C$, $P$, and $\merged{P}$ where $S$
    is normal w.r.t.\ $\pi$ and ${S^\pi \cap J' = \emptyset}$, and assume that
    all of these satisfy the following property:
    \begin{quote}
        ($\lozenge$) for each ${G \in C}$, either ${G^\pi \subseteq (P
        \setminus \merged{P})^\gamma}$ or, for each fact ${H \in G^\pi}$, each
        derivation tree $T$ for $H$ from $E'$ w.r.t.\ ${\Pi \cup \PsameAs}$,
        and each child $t_i$ of the root of $T$, we have ${\pi(F_{t_i}) \in C}$.
    \end{quote}
    Then, property ($\lozenge$) remains preserved after the invocation of
    Algorithm~\ref{alg:checkProvability}.
\end{claim}

\begin{proof}
The proof is by induction on recursion depth of Algorithm
\ref{alg:checkProvability} at which a fact is added to $C$. For the induction
base, ($\lozenge$) remains preserved if the algorithm returns in line
\ref{alg:checkProvability:add-F}.

For the induction step, assume that ($\lozenge$) holds for each fact ${G \in
C}$ different from $F$ after a recursive call in line
\ref{alg:checkProvability:ref:recurse},
\ref{alg:checkProvability:repl:recurse}, or
\ref{alg:checkProvability:rules:recurse}. If the algorithm returns in
line~\ref{alg:checkProvability:allProved},
\ref{alg:checkProvability:ref:allProved},
\ref{alg:checkProvability:repl:allProved},
or~\ref{alg:checkProvability:rules:allProved}, then property 1 of
Claim~\ref{claim:allDisProved} implies ${F^\pi \subseteq (P \setminus
\merged{P})^\gamma}$, so property ($\lozenge$) remains preserved. Otherwise,
consider an arbitrary fact ${H \in F^\pi}$ and an arbitrary derivation tree $T$
for $H$ from $E'$ w.r.t.\ ${\Pi \cup \PsameAs}$. Let ${t_1, \dots, t_n}$ be the
children (if any exist) of the root $\epsilon$ of $T$; since $J$ contains each
fact labelling a node of $T$, we have ${\{ \dtF{t_i}, \dots, \dtF{t_i} \}
\subseteq J' \subseteq J}$. Now let ${F_i = \pi(\dtF{t_i})}$; by the definition
of r-materialisation, we have ${\{ F_1, \dots, F_n \} \subseteq I}$. Moreover,
for each ${1 \leq i \leq n}$, we have ${F_i \in J'}$ and ${S^\pi \cap J' =
\emptyset}$, which imply ${F_i \not\in S^\pi}$; moreover, $S$ is normal w.r.t.\
$\pi$, so ${F_i \not\in S}$ as well. Finally, we clearly have
${\pi(\dtR{\epsilon}\dtS{\epsilon}) = \pi(\dtR{\epsilon})\pi(\dtS{\epsilon})}$,
and so ${\head{\pi(\dtR{\epsilon})}\pi(\dtS{\epsilon}) = F}$ and
${\body{\pi(\dtR{\epsilon})}\pi(\dtS{\epsilon}) = \{ F_1, \dots, F_n \}
\subseteq I \setminus S}$. We next consider the forms of $\dtR{\epsilon}$.
\begin{itemize}
    \item Assume $\dtR{\epsilon}$ is of the form \eqref{eq:eq-ref}, so ${n =
    1}$. Fact ${F_1}$ is eventually considered in
    line~\ref{alg:checkProvability:ref:iterate-G}, so, due to the recursive
    call in line~\ref{alg:checkProvability:ref:recurse}, we have ${F_1 \in C}$,
    as required.

    \item Assume $\dtR{\epsilon}$ is of the form
    \eqref{eq:eq-repl1}--\eqref{eq:eq-repl3}; thus, ${n = 2}$, ${F_1 = F}$, and
    ${F_2 = c \sameAs c}$ for some constant $c$. Fact ${F_1 = F}$ is added to
    $C$ in line~\ref{alg:checkProvability:add-F}. Moreover, by assumption
    ($\ast$) on the shape of $T$, we have ${F_2 = s \sameAs t}$ with ${s \neq
    t}$; since ${\pi(s) = \pi(t) = c}$, we have ${|c^\pi| >1}$. Thus, due to
    the recursive call in line~\ref{alg:checkProvability:repl:recurse}, we have
    ${F_2 \in C}$, as required.

    \item Assume ${\dtR{\epsilon} \in \Pi}$. Then, ${\pi(\dtR{\epsilon}) \in
    \pi(\Pi)}$, so $\pi(\dtR{\epsilon})$ and $\pi(\dtS{\epsilon})$ are
    eventually considered in lines~\ref{alg:checkProvability:rules:1}
    and~\ref{alg:checkProvability:rules:2}; hence, due to the recursive call in
    line~\ref{alg:checkProvability:rules:recurse}, we have ${F_i \in C}$ for
    each ${1 \leq i \leq n}$, as required. \qedhere
\end{itemize}
\end{proof}

Calls in line \ref{alg:BF-sameAs:checkProvability} ensure another property on
$C$, $P$, $\merged{P}$, and $S$.

\begin{claim}\label{claim:proved-disproved}
    The following properties hold after each line of Algorithm
    \ref{alg:BF-sameAs}:
    \begin{enumerate}
        \item property ($\lozenge$) is satisfied;

        \item ${(P \setminus \merged{P})^\gamma = C^\pi \cap J'}$;

        \item ${\gamma(c) = \min \eqclass{C^\pi \cap J'}{c}}$ for each constant
        $c$; and

        \item ${S^\pi \cap J' = \emptyset}$.

        \item For each fact ${F \in O}$, we have ${F^\pi \not\subseteq J'}$.

        \item ${D \subseteq C}$.
    \end{enumerate}
\end{claim}

\begin{proof}
The proof is by induction on the number of iterations of the loop in
lines~\ref{alg:BF-sameAs:extract-F}--\ref{alg:BF-sameAs:rules:processed}. For
the induction base, we have ${S = C = P = O = \emptyset}$ in line
\ref{alg:BF-sameAs:initialise}, so properties 1--5 clearly hold initially. For
the induction step, assume that all properties hold before
line~\ref{alg:BF-sameAs:checkProvability}. Due to property 4 and
Claim~\ref{claim:checkProvability}, property 1 remains preserved after
line~\ref{alg:BF-sameAs:checkProvability}; hence, we next consider properties
2--6.

\smallskip

(Property 2) Let $K$ and $L$ be as stated in Claim~\ref{claim:saturate}; note
that property 2 of Claim~\ref{claim:saturate} is equivalent to ${(P \setminus
\merged{P})^\gamma = L}$. We first show ${(P \setminus \merged{P})^\gamma
\subseteq C^\pi \cap J'}$. Since ${K \subseteq J'}$, we clearly have ${J' =
(\Pi \cup \PsameAs)^\infty(K \cup E')}$. Moreover, for each ${F \in (P
\setminus \merged{P})^\gamma}$ we have ${F \in L}$, so by the definition of $L$
there exists a derivation tree $T$ for $F$ from ${K \cup E'}$ w.r.t.\ ${P \cup
\PsameAs}$ such that ${\dtF{t} \in C^\pi}$ holds for each node $t$ of $T$; but
then, we clearly have ${F \in C^\pi \cap J'}$. We next prove ${C^\pi \cap J'
\subseteq (P \setminus \merged{P})^\gamma}$ by induction on the height $h$ of a
fact ${F \in C^\pi \cap J'}$ w.r.t.\ $E'$ and ${\Pi \cup \PsameAs}$.
\begin{itemize}
    \item If ${h = 0}$, then ${F \in E'}$; since ${F \in C^\pi}$, by the
    definition of $L$ we have ${F \in L}$; but then, ${F \in (P \setminus
    \merged{P})^\gamma}$ as well.

    \item Assume that the claim holds for each fact in ${C^\pi \cap J'}$ whose
    height w.r.t.\ $E'$ and ${\Pi \cup \PsameAs}$ is at most $h$, and consider
    an arbitrary fact ${F \in C^\pi \cap J'}$ with height $h+1$; let $T$ be the
    corresponding derivation tree for $F$. Moreover, assume that ${F \not\in (P
    \setminus \merged{P})^\gamma}$; then, ${F \in C^\pi}$ implies ${\pi(F) \in
    C}$; hence, property ($\lozenge$) ensures that, for each child $t_i$ of the
    root of $T$, we have ${\pi(\dtF{t_i}) \in C}$, which is equivalent to
    ${\dtF{t_i} \in C^\pi}$. Now the height of each $\dtF{t_i}$ w.r.t.\ $E'$
    and ${\Pi \cup \PsameAs}$ is at most $h$ so, by the induction assumption,
    we have ${\dtF{t_i} \in (P \setminus \merged{P})^\gamma = L}$. The latter
    ensures that, for each $\dtF{t_i}$, there exists a derivation tree $T_i$ in
    which each node is labelled by a fact contained in $C^\pi$. Let $T'$ be the
    derivation tree in which the root $\epsilon$ is labelled with the same
    fact, rule, and substitution as in $T$, and each $T_i$ is a subtree of
    $\epsilon$. Clearly, $T'$ is a derivation tree for $F$ from $E'$ w.r.t.\
    ${\Pi \cup \PsameAs}$ in which each node is labelled by a fact contained in
    $C^\pi$; thus, by the definition of $L$, we have ${F \in L = (P \setminus
    \merged{P})^\gamma}$, as required.
\end{itemize}

\smallskip

(Property 3) This property follows directly from property 1 of
Claim~\ref{claim:saturate} and property 2 of Claim~\ref{claim:proved-disproved}.

\smallskip

(Property 4) Assume that some fact $G$ is added to $S$ in
line~\ref{alg:BF-sameAs:disproved}. Then ${\mathsf{allDisproved}(G) = \true}$,
which by property 2 of Claim~\ref{claim:allDisProved} implies ${G^\pi \cap (P
\setminus \merged{P})^\gamma = \emptyset}$. Property 2 of
Claim~\ref{claim:proved-disproved} holds at this point, so we have ${G^\pi \cap
C^\pi \cap J' = \emptyset}$. Finally,
lines~\ref{alg:BF-sameAs:checkProvability} and~\ref{alg:checkProvability:add-F}
ensure ${G \in C}$, so we have ${G^\pi \subseteq C^\pi}$; thus, ${G^\pi \cap J'
= \emptyset}$, and so adding $G$ to $S$ preserves property 4.

\smallskip

(Property 5) Assume that some fact $F$ is added to $O$ in
line~\ref{alg:BF-sameAs:rules:processed}. Then ${\mathsf{allProved}(F) =
\false}$, which by property 1 of Claim~\ref{claim:allDisProved} implies ${F \in
S}$ or ${F^\pi \not\subseteq (P \setminus \merged{P})^\gamma}$. In the former
case, ${F^\pi \not\subseteq J'}$ holds directly from property 4. In the latter
case, property 2 of Claim~\ref{claim:proved-disproved} holds at this point, so
we have ${F^\pi \not\subseteq C^\pi \cap J'}$; moreover,
lines~\ref{alg:BF-sameAs:checkProvability} and~\ref{alg:checkProvability:add-F}
ensure ${F \in C}$, which implies ${F^\pi \subseteq C^\pi}$; this, in turn,
implies ${F^\pi \not\subseteq J'}$. Consequently, adding $F$ to $O$ preserves
property 5.

\smallskip

(Property 6) Each fact $F$ extracted from $D$ in
line~\ref{alg:BF-sameAs:extract-F} is passed in
line~\ref{alg:BF-sameAs:checkProvability} to
Algorithm~\ref{alg:checkProvability}, which in turn ensures that $F$ is added
to $C$ in line~\ref{alg:checkProvability:add-F}.
\end{proof}

We next show that set $D$ contains each fact that needs to be deleted, and each
fact that contains a constant whose representative changes as a result of the
update.

\begin{claim}\label{claim:D-complete}
    For each fact ${F \in J \setminus J'}$, the following two
    properties hold in line~\ref{alg:BF-sameAs:propagateChanges}:
    \begin{enumerate}
        \item ${\pi(F) \in D}$, and

        \item if ${F = s \sameAs t}$ with ${s \neq t}$, then $D$ contains each
        fact ${G \in I}$ such that ${\pi(s) \in \voc{G}}$ and ${G^\pi
        \not\subseteq J'}$.
    \end{enumerate}
\end{claim}

\begin{proof} Consider an arbitrary fact ${F \in J \setminus J'}$.

(Property 1) We prove the claim by induction on the height $h$ of $F$ w.r.t.\
$E$ and ${\Pi \cup \PsameAs}$; the notion of the height of $F$ is correctly
defined because ${F \in J}$. For the induction base, assume ${h = 0}$; now ${F
\in J}$ implies ${F \in E}$; moreover, ${F \not\in J'}$ implies ${F \not\in
E'}$; thus, ${F \in E^-}$, and so $\pi(F)$ is added to $D$ in lines
\ref{alg:BF-sameAs:update-E-D:1}--\ref{alg:BF-sameAs:update-E-D:2}. For the
induction step, assume that the claim holds for each fact in ${J \setminus J'}$
whose height w.r.t.\ $E$ and ${\Pi \cup \PsameAs}$ is at most $h$, and assume
that the height of $F$ w.r.t.\ $E$ and ${\Pi \cup \PsameAs}$ is $h+1$. Let $T$
be a corresponding derivation tree for $F$ from $E$ w.r.t.\ ${\Pi \cup
\PsameAs}$; let ${t_1, \dots, t_n}$ be the children of the root $\epsilon$ of
$T$; and let ${F_i = \pi(\dtF{t_i})}$ for each ${1 \leq i \leq n}$. Moreover,
let $N$ contain precisely each $F_i$, ${1 \leq i \leq n}$, such that ${F_i \in
D}$ and ${F_i^\pi \not\subseteq J}$. Since ${F \not\in J'}$, some $j$ with ${1
\leq j \leq n}$ exists such that ${\dtF{t_j} \not\in J'}$; moreover, $T$ is a
derivation tree for $F$ from $E$ w.r.t.\ ${\Pi \cup \PsameAs}$, so ${\dtF{t_j}
\in J}$ and the height of $\dtF{t_j}$ is at most $h$; but then, we have
${\pi(\dtF{t_j}) = F_j \in D}$ by the induction hypothesis, and so we also have
${F_j \in N}$---that is, ${N \neq \emptyset}$. Each fact in $D$ is eventually
considered in line~\ref{alg:BF-sameAs:extract-F}; thus, let $F'$ be the fact
from $N$ that is consider first. At that point, we have ${O \cap N =
\emptyset}$ because facts are added to added to $O$ in
line~\ref{alg:BF-sameAs:rules:processed} only after they have been considered;
hence, ${F_i \in I \setminus O}$ holds at this point for each ${1 \leq i \leq
n}$. Furthermore, ${F' \in D \subseteq C}$ implies ${(F')^\pi \subseteq
C^\pi}$; but then, ${(F')^\pi \not\subseteq J'}$ and property 2 of
Claim~\ref{claim:proved-disproved} imply ${(F')^\pi \not\subseteq (P \setminus
\merged{P})^\gamma}$; thus, property 1 of Claim~\ref{claim:allDisProved}
ensures we have ${\mathsf{allProved}(F') = \false}$ and so the check in
line~\ref{alg:BF-sameAs:notAllProved} passes. We next consider the possible
forms of the rule $\dtR{\epsilon}$.
\begin{itemize}
    \item Assume that $\dtR{\epsilon}$ is
    \eqref{eq:eq-repl1}--\eqref{eq:eq-repl3}. Then, we clearly have ${\pi(F) =
    F_1}$; fact $\dtF{t_2}$ is of the form ${\dtF{t_2} = s \sameAs t}$ with ${s
    \neq t}$ and ${c = \pi(s) = \pi(t)}$; and ${c \in \voc{F_1}}$. We have two
    possible ways to choose $F'$. If ${F' = F_1}$, then ${\pi(F) = F_1 = F' \in
    D}$ holds. If ${F' = F_2}$, then ${s \neq t}$ by assumption ($\ast$) on the
    shape of $T$, so ${|c^\pi| > 1}$ and the check in
    line~\ref{alg:BF-sameAs:repl} passes; furthermore, due to ${F_1 \in I
    \setminus O}$, we eventually consider fact ${G = F_1 = \pi(F)}$ in
    line~\ref{alg:BF-sameAs:repl:iterate-G} and add it to $D$ in
    line~\ref{alg:BF-sameAs:repl:add-G}.

    \item Assume that $\dtR{\epsilon}$ is \eqref{eq:eq-ref}. Then, $F$ is of
    the form ${s \sameAs s}$ so ${\pi(F) = c \sameAs c}$ for ${c = \pi(s)}$;
    clearly, we have ${c \in \voc{F'}}$ and ${F' = F_1}$. But then, $\pi(F)$ is
    added to $D$ in line~\ref{alg:BF-sameAs:ref}.

    \item Assume that ${\dtR{\epsilon} \in \Pi}$. We clearly have
    ${\pi(\dtR{\epsilon}\dtS{\epsilon}) =
    \pi(\dtR{\epsilon})\pi(\dtS{\epsilon})}$; therefore, we have ${\pi(F) =
    \pi(\head{\dtR{\epsilon}\dtS{\epsilon}}) =
    \head{\pi(\dtR{\epsilon})}\pi(\dtS{\epsilon})}$ and
    ${\pi(\body{\dtR{\epsilon}\dtS{\epsilon}}) = \{ F_1, \dots, F_n \} =
    \body{\pi(\dtR{\epsilon})}\pi(\dtS{\epsilon}) \subseteq I \setminus O}$.
    Moreover, we clearly have ${\pi(\dtR{\epsilon}) \in \pi(\Pi)}$. Finally,
    let $i$ be the smallest integer with ${1 \leq i \leq n}$ such that ${F_i =
    F'}$, and let $Q$ be annotated query \eqref{eq:seminaive-annotated-query}
    obtained from $\pi(\dtR{\epsilon})$ for that $i$; clearly, the way in which
    we chose $i$ ensures ${F_j \neq F'}$ for each $j$ with ${1 \leq j < i}$.
    All of these observations ensure together that ${\langle
    \pi(\dtR{\epsilon}),Q,\sigma) \rangle \in \matchBody{\pi(\Pi)}{F'}}$ is
    considered in line~\ref{alg:BF-sameAs:rules:1}, and that
    $\pi(\dtS{\epsilon})$ is considered in line~\ref{alg:BF-sameAs:rules:2};
    consequently, $\pi(F)$ is added to $D$ in line
    \ref{alg:BF-sameAs:rules:derive}.
\end{itemize}

(Property 2) Assume that $F$ is of the form ${F = s \sameAs t}$ with ${s \neq
t}$, let ${c = \pi(s) = \pi(t)}$, and let ${F' = \pi(F)}$. Property 1 of this
claim ensures ${F' = c \sameAs c \in D \subseteq C}$, and so we have ${(F')^\pi
\subseteq C^\pi}$; but then, together with ${F \not\in J'}$, property 2 of
Claim~\ref{claim:proved-disproved} ensures ${(F')^\pi \not\subseteq (P
\setminus \merged{P})^\gamma}$; finally, property 1 of
Claim~\ref{claim:allDisProved} ensures ${\mathsf{allProved}(F') = \false}$.
Fact $F'$ is eventually processed in line~\ref{alg:BF-sameAs:extract-F}, and by
the previous discussion the check in line~\ref{alg:BF-sameAs:notAllProved}
passes. Moreover, ${s \neq t}$ implies ${|c^\pi| >1}$, so the check in
line~\ref{alg:BF-sameAs:repl} passes as well. Now consider an arbitrary fact
${G \in I}$ such that ${c \in \voc{G}}$ and ${G^\pi \not\subseteq J'}$;
property 5 of Claim~\ref{claim:proved-disproved} ensures ${G \not\in O}$, and
therefore $G$ is added to $D$ in line~\ref{alg:BF-sameAs:repl:add-G}.
\end{proof}

We next show that Algorithm~\ref{alg:BF-sameAs} correctly updates $I$ to $I'$.

\begin{claim}\label{claim:BF-sameAs:correctness}
    Algorithm~\ref{alg:BF-sameAs} updates set $I$ to $I'$.
\end{claim}

\begin{proof}
Property 6 of Claim~\ref{claim:proved-disproved} and property 1 of
Claim~\ref{claim:D-complete} clearly ensure that \eqref{eq:all-checked} holds.
Furthermore, property 2 of Claim~\ref{claim:proved-disproved} clearly ensures
that \eqref{eq:all-contained} holds.
\begin{align}
    J \setminus J'                  & \subseteq D^\pi \subseteq C^\pi \label{eq:all-checked} \\
    (P \setminus \merged{P})^\gamma & \subseteq J' \subseteq J \label{eq:all-contained}
\end{align}
 For convenience we recapitulate the definitions of $\pi(c)$, $\pi'(c)$, and
$\gamma(c)$; note that \eqref{eq:gamma-c} follows immediately from properties 2
and 3 of Claim~\ref{claim:proved-disproved}. Finally, \eqref{eq:all-contained},
\eqref{eq:pi-P-c}, and \eqref{eq:gamma-c} clearly imply
\eqref{eq:pi-P-of-gamma}.
\begin{align}
    \pi(c)                                  & = \min \eqclass{J}{c} \label{eq:pi-c} \\
    \pi'(c)                                 & = \min \eqclass{J'}{c} \label{eq:pi-P-c} \\
    \gamma(c)                               & = \min \eqclass{(P \setminus \merged{P})^\gamma}{c} \label{eq:gamma-c} \\
    \pi'((P \setminus \merged{P})^\gamma)   & = \pi'(P \setminus \merged{P}) \label{eq:pi-P-of-gamma}
\end{align}

Before proceeding, we prove several useful properties. Consider an arbitrary
constant $c$ with ${\pi(c) = c}$; by \eqref{eq:all-contained} and
\eqref{eq:pi-c}--\eqref{eq:gamma-c}, we clearly have ${\pi'(c) = c}$ and
${\gamma(c) = c}$. Thus, for each fact $F$ with ${\pi(F) = F}$, we have
${\pi'(F) = F}$ and ${\gamma(F) = F}$, which ensures the following properties:
\begin{align}
    \begin{array}{r@{}l@{\;}l@{\qquad}r@{}l@{\qquad}r@{}l}
        F \in I & \text{ iff } F \in J,     &               & F \in I'  & \text{ iff } F \in J', & F \in (P \setminus \merged{P})^\gamma   & \text{ iff } F \in P \setminus \merged{P}, \\
        F \in D & \text{ iff } F \in D^\pi, & \text{ and }  & F \in C   & \text{ iff } F \in C^\pi.
    \end{array} \label{eq:norm-F}
\end{align}

\smallskip

We next show that
lines~\ref{alg:propagateChanges:pi-1}--\ref{alg:propagateChanges:pi-2} update
$\pi$ to $\pi'$. To this end, consider arbitrary constants $c$ and $d$ with
${\pi(d) = c}$, and let ${F = c \sameAs c}$. Set $F^\pi$ clearly contains each
triple of the form ${d \sameAs e \in J}$, which, together
with~\eqref{eq:all-contained}, implies
\begin{align}
    \eqclass{F^\pi \cap (P \setminus \merged{P})^\gamma}{d} = \eqclass{(P \setminus \merged{P})^\gamma}{d}, \qquad \eqclass{F^\pi \cap J'}{d} = \eqclass{J'}{d}, \qquad \text{and} \qquad \eqclass{F^\pi \cap J}{d} = \eqclass{J}{d}. \label{eq:classes}
\end{align}
We now consider two possible cases.
\begin{itemize}
    \item Assume that ${F \in C}$. Thus, ${F^\pi \subseteq C^\pi}$ holds, so
    property 2 of Claim~\ref{claim:proved-disproved} ensures ${F^\pi \cap (P
    \setminus \merged{P})^\gamma = F^\pi \cap J' = V}$. But then,
    \eqref{eq:classes} imply ${\eqclass{V}{d} = \eqclass{J'}{d} = \eqclass{(P
    \setminus \merged{P})^\gamma}{d}}$. Finally, \eqref{eq:pi-P-c} and
    \eqref{eq:gamma-c} imply ${\pi'(d) = \gamma(d)}$.

    \item Assume that ${F \not\in C}$. We thus have ${F^\pi \cap C^\pi =
    \emptyset}$; but then, ${J \setminus J' \subseteq C^\pi}$ implies ${F^\pi
    \cap (J \setminus J') = \emptyset}$, which then implies ${F^\pi \cap J =
    F^\pi \cap J'}$. Finally, \eqref{eq:pi-c}, \eqref{eq:pi-P-c}, and
    \eqref{eq:classes} together imply ${\pi'(d) = \pi(d)}$.
\end{itemize}

\smallskip

We next prove ${I \setminus I' = D \setminus (P \setminus \merged{P})}$ and
hence show that line~\ref{alg:propagateChanges:delete-unproved} correctly
deletes the relevant facts. To this end, we next consider each side of the
inclusion.
\begin{itemize}
    \item Assume that ${F \in I \setminus I'}$. Then ${F \in I}$ implies
    ${\pi(F) = F}$, so by \eqref{eq:norm-F} we have ${F \in J \setminus J'}$.
    By \eqref{eq:all-checked} we have ${F \in D^\pi \subseteq C^\pi}$, and by
    \eqref{eq:norm-F} we have ${F \in D \subseteq C}$. Moreover, ${F \not\in
    J'}$ and property 2 of Claim~\ref{claim:proved-disproved} imply ${F \not\in
    (P \setminus \merged{P})^\gamma}$, which by \eqref{eq:norm-F} implies ${F
    \not\in P \setminus \merged{P}}$. Consequently, we have ${F \in D \setminus
    (P \setminus \merged{P})}$.

    \item Assume that ${F \in D \setminus (P \setminus \merged{P})}$. Then ${D
    \subseteq I}$ implies ${F \in I}$, so ${\pi(F) = F}$. Also, ${F \not\in P
    \setminus \merged{P}}$ and \eqref{eq:norm-F} imply ${F \not\in (P \setminus
    \merged{P})^\gamma}$. But then, property 2 of
    Claim~\ref{claim:proved-disproved} ensures ${F \not\in C^\pi \cap J'}$. Due
    to ${D \subseteq C}$ and \eqref{eq:norm-F}, we have ${F \in C^\pi}$; thus,
    ${F \not\in J'}$, so by \eqref{eq:norm-F} we have ${F \not\in I'}$.
    Consequently, we have ${F \in I \setminus I'}$.
\end{itemize}

\smallskip

We finally prove that $I' = [I \setminus (I \setminus I')] \cup \pi'(P
\setminus \merged{P})$ and hence show that
line~\ref{alg:propagateChanges:add-proved} correctly adds the relevant facts;
please remember that, due to updates in
lines~\ref{alg:propagateChanges:pi-1}--\ref{alg:propagateChanges:pi-2}, mapping
$\pi$ actually contains $\pi'$ in line~\ref{alg:propagateChanges:add-proved}.
\begin{itemize}
    \item Assume that ${F \in [I \setminus (I \setminus I')] \cup \pi'(P
    \setminus \merged{P})}$. We consider two cases.
    \begin{itemize}
        \item Assume that ${F \in I \setminus (I \setminus I')}$. Thus, ${F \in
        I}$ and ${F \not\in I \setminus I'}$; but then, we have ${F \in I'}$,
        as required.

        \item Assume that ${F \in \pi'(P \setminus \merged{P})}$. Then, some
        ${G \in (P \setminus \merged{P})^\gamma}$ exists such that ${\pi'(G) =
        F}$. By property 2 of Claim~\ref{claim:proved-disproved}, we have ${G
        \in J'}$; but then, we have ${\pi'(G) = F \in I'}$, as required.
    \end{itemize}

    \item Assume that ${F \in I'}$ and ${F \not\in I \setminus (I \setminus
    I')}$. Thus, ${F \not\in I}$, but clearly ${F \in J' \subseteq J}$. Due to
    the latter, some ${G \in I}$ exists such that ${\pi(F) = G}$; clearly, ${F
    \neq G}$ and ${G^\pi \not\subseteq J}$. Since ${G \in I}$, we have ${\pi(G)
    = G}$; thus, by \eqref{eq:norm-F} we have ${\pi'(G) = G}$. Moreover, ${F
    \in I'}$ implies ${\pi'(F) = F}$. Consequently, distinct constants ${a \in
    \voc{F}}$ and ${b \in \voc{G}}$ exist such that ${a \sameAs b \in J
    \setminus J'}$; but then, property 2 of Claim~\ref{claim:D-complete} and
    ${G^\pi \not\subseteq J}$ ensure that ${G \in D \subseteq C \subseteq
    C^\pi}$, which ensures ${F \in C^\pi}$. Since ${F \in J'}$, by property 2
    of Claim~\ref{claim:proved-disproved} we have ${F \in (P \setminus
    \merged{P})^\gamma}$; but then, by \eqref{eq:pi-P-of-gamma} we have ${F \in
    \pi'(P \setminus \merged{P})}$, as required. \qedhere
\end{itemize}
\end{proof}

We next show that Algorithm~\ref{alg:BF-sameAs} does not repeat derivations.

\begin{claim}\label{claim:B-F-delete:nonrepetition}
    Each combination of a rule $r$ and a substitution $\tau$ is considered at
    most once in line~\ref{alg:BF-sameAs:rules:derive}.
\end{claim}

\begin{proof}
Assume that a rule ${r \in \Pi}$ and substitution $\tau$ exist that are
considered in line~\ref{alg:BF-sameAs:rules:derive} twice, when (not
necessarily distinct) facts $F$ and $F'$ are extracted from $D$. Moreover, let
$B_i$ and $B_{i'}$ be the body atoms of $r$ that $\tau$ matches to $F$ and
$F'$---that is, ${F = B_i\tau}$ and ${F' = B_{i'}\tau}$. Finally, let $Q'$ be
the annotated query considered in line \ref{alg:BF-sameAs:rules:1} when atom
$B_{i'}$ of $r$ is matched to $F'$. We have the following possibilities.
\begin{itemize}
    \item Assume that $F = F'$. Then, $B_i$ and $B_{i'}$ must be distinct, so
    w.l.o.g.\ assume that ${i \leq i'}$. But then, query $Q'$ contains atom
    $B_i^{\neq}$, so $\tau$ cannot be returned in
    line~\ref{alg:BF-sameAs:rules:2} when evaluating $Q'$.

    \item Assume that ${F \neq F'}$ and that, w.l.o.g.\ $F$ is extracted from
    $D$ before $F'$. Then, we have ${F \in O}$ due to
    line~\ref{alg:BF-sameAs:rules:processed}, and therefore we have ${F \not\in
    I \setminus O}$; consequently, $\tau$ cannot be returned in
    line~\ref{alg:BF-sameAs:rules:2} when evaluating $Q'$. \qedhere
\end{itemize}
\end{proof}

%% file: paper.bbl
\begin{thebibliography}{}

\bibitem[\protect\citeauthoryear{Abiteboul \bgroup \em et al.\egroup
  }{1995}]{abiteboul95foundation}
S.~Abiteboul, R.~Hull, and V.~Vianu.
\newblock {\em {Foundations of Databases}}.
\newblock Addison Wesley, 1995.

\bibitem[\protect\citeauthoryear{Aref}{2010}]{DBLP:conf/iclp/Aref10}
Molham Aref.
\newblock {Datalog for Enterprise Software: from Industrial Applications to
  Research (Invited Talk)}.
\newblock In {\em Tech. Comm. ICLP}, volume~7, page~1, 2010.

\bibitem[\protect\citeauthoryear{Baader and Nipkow}{1998}]{baader98term}
F.~Baader and T.~Nipkow.
\newblock {\em {Term Rewriting and All That}}.
\newblock CUP, 1998.

\bibitem[\protect\citeauthoryear{Bishop \bgroup \em et al.\egroup
  }{2011}]{DBLP:journals/semweb/BishopKOPTV11}
Barry Bishop, Atanas Kiryakov, Damyan Ognyanoff, Ivan Peikov, Zdravko Tashev,
  and Ruslan Velkov.
\newblock {OWLIM}: A family of scalable semantic repositories.
\newblock {\em Semantic Web}, 2(1):33--42, 2011.

\bibitem[\protect\citeauthoryear{de Kleer}{1986}]{DBLP:journals/ai/Kleer86}
Johan de~Kleer.
\newblock {An Assumption-Based TMS}.
\newblock {\em Artificial Intelligence}, 28(2):127--162, 1986.

\bibitem[\protect\citeauthoryear{Dewan \bgroup \em et al.\egroup
  }{1992}]{DBLP:journals/jiis/DewanOSWS92}
H.~M. Dewan, D.~Ohsie, S.~J. Stolfo, O.~Wolfson, and S.~Da Silva.
\newblock {Incremental Database Rule Processing In PARADISER}.
\newblock {\em Journal of Intelligent Information Systems}, 1(2):177--209,
  1992.

\bibitem[\protect\citeauthoryear{Doyle}{1979}]{DBLP:journals/ai/Doyle79}
Jon Doyle.
\newblock {A Truth Maintenance System}.
\newblock {\em Artificial Intelligence}, 12(3):231--272, 1979.

\bibitem[\protect\citeauthoryear{Goasdou{\'{e}} \bgroup \em et al.\egroup
  }{2013}]{DBLP:conf/edbt/GoasdoueMR13}
Fran{\c{c}}ois Goasdou{\'{e}}, Ioana Manolescu, and Alexandra Roatis.
\newblock Efficient query answering against dynamic {RDF} databases.
\newblock In {\em Proc. EDBT}, pages 299--310. {ACM}, 2013.

\bibitem[\protect\citeauthoryear{Grosof \bgroup \em et al.\egroup
  }{2003}]{GHVD03}
B.~N. Grosof, I.~Horrocks, R.~Volz, and S.~Decker.
\newblock {Description Logic Programs: Combining Logic Programs with
  Description Logic}.
\newblock In {\em Proc.\ WWW}, pages 48--57, 2003.

\bibitem[\protect\citeauthoryear{Guo \bgroup \em et al.\egroup
  }{2005}]{DBLP:journals/ws/GuoPH05}
Y.~Guo, Z.~Pan, and J.~Heflin.
\newblock {LUBM: A benchmark for OWL knowledge base systems}.
\newblock {\em Journal of Web Semantics}, 3(2--3):158--182, 2005.

\bibitem[\protect\citeauthoryear{Gupta \bgroup \em et al.\egroup
  }{1993}]{DBLP:conf/sigmod/GuptaMS93}
A.~Gupta, I.~S. Mumick, and V.~S. Subrahmanian.
\newblock {Maintaining Views Incrementally}.
\newblock In {\em Proc. SIGMOD}, pages 157--166. ACM, 1993.

\bibitem[\protect\citeauthoryear{Horrocks \bgroup \em et al.\egroup
  }{2004}]{SWRL}
I.~Horrocks, P.~F. Patel-Schneider, H.~Boley, S.~Tabet, B.~Grosof, and M.~Dean.
\newblock {SWRL: A Semantic Web Rule Language Combining OWL and RuleML, W3C
  Member Submission}, 2004.

\bibitem[\protect\citeauthoryear{Ma \bgroup \em et al.\egroup
  }{2006}]{DBLP:conf/esws/MaYQXPL06}
L.~Ma, Y.~Yang, Z.~Qiu, G.~T. Xie, Y.~Pan, and S.~Liu.
\newblock {Towards a Complete OWL Ontology Benchmark}.
\newblock In {\em Proc.\ ESWC}, pages 125--139, 2006.

\bibitem[\protect\citeauthoryear{Motik \bgroup \em et al.\egroup
  }{2009}]{owl2-profiles}
B.~Motik, B.~{Cuenca Grau}, I.~Horrocks, Z.~Wu, A.~Fokoue, and C.~Lutz.
\newblock {OWL 2 Web Ontology Language: Profiles, W3C Recommendation}, October
  27 2009.

\bibitem[\protect\citeauthoryear{Motik \bgroup \em et al.\egroup
  }{2014}]{mnpho14parallel-materialisation-RDFox}
Boris Motik, Yavor Nenov, Robert Piro, Ian Horrocks, and Dan Olteanu.
\newblock {Parallel Materialisation of Datalog Programs in Centralised,
  Main-Memory RDF Systems}.
\newblock In {\em Proc.\ AAAI}, 2014.

\bibitem[\protect\citeauthoryear{Motik \bgroup \em et al.\egroup
  }{2015a}]{mnph15owl-sameAs-rewriting}
Boris Motik, Yavor Nenov, Robert Piro, and Ian Horrocks.
\newblock {Handling owl:sameAs via Rewriting}.
\newblock In {\em Proc.\ AAAI}, 2015.

\bibitem[\protect\citeauthoryear{Motik \bgroup \em et al.\egroup
  }{2015b}]{mnph15incremental-BF}
Boris Motik, Yavor Nenov, Robert Piro, and Ian Horrocks.
\newblock {Incremental Update of Datalog Materialisation: the Backward/Forward
  Algorithm}.
\newblock In {\em Proc.\ AAAI}, 2015.

\bibitem[\protect\citeauthoryear{Nicolas and
  Yazdanian}{1983}]{DBLP:conf/ifip/NicolasY83}
J.-M. Nicolas and K.~Yazdanian.
\newblock {An Outline of BDGEN: A Deductive DBMS}.
\newblock In {\em Proc.\ IFIP}, pages 711--717, 1983.

\bibitem[\protect\citeauthoryear{Nieuwenhuis and
  Rubio}{2001}]{NieuwenhuisRubio:HandbookAR:paramodulation:2001}
R.~Nieuwenhuis and A.~Rubio.
\newblock {Paramodulation-Based Theorem Proving}.
\newblock In A.~Robinson and A.~Voronkov, editors, {\em Handbook of Automated
  Reasoning}, volume~I, chapter~7, pages 371--443. Elsevier Science, 2001.

\bibitem[\protect\citeauthoryear{Stocker and
  Smith}{2008}]{DBLP:conf/owled/StockerS08}
Markus Stocker and Michael Smith.
\newblock {Owlgres: A Scalable OWL Reasoner}.
\newblock In {\em Proc.\ OWLED}, 2008.

\bibitem[\protect\citeauthoryear{Urbani \bgroup \em et al.\egroup
  }{2012}]{DBLP:journals/ws/UrbaniKMHB12}
J.~Urbani, S.~Kotoulas, J.~Maassen, F.~van Harmelen, and H.~E. Bal.
\newblock {WebPIE: A Web-scale Parallel Inference Engine using MapReduce}.
\newblock {\em Journal of Web Semantics}, 10:59--75, 2012.

\bibitem[\protect\citeauthoryear{Urbani \bgroup \em et al.\egroup
  }{2013}]{DBLP:conf/semweb/UrbaniMJHB13}
J.~Urbani, A.~Margara, C.~J.~H. Jacobs, F.~{van Harmelen}, and H.~E. Bal.
\newblock {DynamiTE: Parallel Materialization of Dynamic RDF Data}.
\newblock In {\em Proc.\ ISWC}, volume 8218, pages 657--672. Springer, 2013.

\bibitem[\protect\citeauthoryear{Wu \bgroup \em et al.\egroup
  }{2008}]{DBLP:conf/icde/WuEDCKAS08}
Z.~Wu, G.~Eadon, S.~Das, E.~I. Chong, V.~Kolovski, M.~Annamalai, and
  J.~Srinivasan.
\newblock {Implementing an Inference Engine for RDFS/OWL Constructs and
  User-Defined Rules in Oracle}.
\newblock In {\em Proc. ICDE}, pages 1239--1248. IEEE, 2008.

\bibitem[\protect\citeauthoryear{Zhou \bgroup \em et al.\egroup
  }{2013}]{DBLP:conf/www/ZhouGHWB13}
Y.~Zhou, B.~{Cuenca Grau}, I.~Horrocks, Z.~Wu, and J.~Banerjee.
\newblock {Making the most of your triple store: query answering in OWL 2 using
  an RL reasoner}.
\newblock In {\em Proc.\ WWW}, pages 1569--1580, 2013.

\end{thebibliography}
